\documentclass[fleqn]{llncs}
\usepackage{version}
\usepackage{nuc}

\title{Instruction Sequences and\\ Non-uniform Complexity Theory%
       \thanks{This research was partly carried out in the framework of
               the  Jacquard-project Symbiosis, which is funded by the
               Netherlands Organisation for Scientific Research (NWO).}}
\author{J.A. Bergstra \and C.A. Middelburg}
\institute{Informatics Institute, Faculty of Science,
           University of Amsterdam, \\
           Science Park~904, 1098~XH Amsterdam, the Netherlands \\
           \email{J.A.Bergstra@uva.nl,C.A.Middelburg@uva.nl}}

\begin{document}

\maketitle

\begin{abstract}
We develop theory concerning non-uniform complexity in a setting in
which the notion of single-pass instruction sequence considered in
program algebra is the central notion.
We define counterparts of the complexity classes \PTpoly\ and \NPTpoly\
and formulate a counterpart of the complexity theoretic conjecture that
$\NPT \not\subseteq \PTpoly$.
In addition, we define a notion of completeness for the counterpart of
\NPTpoly\ using a non-uniform reducibility relation and formulate
complexity hypotheses which concern restrictions on the instruction
sequences used for computation.
We think that the theory developed opens up an additional way of
investigating issues concerning non-uniform complexity.
\begin{keywords}
single-pass instruction sequence, non-uniform complexity, \linebreak
non-uniform super-polynomial complexity hypothesis,
super-polynomial feature elimination complexity
\end{keywords}%
\begin{classcode}
F.1.1, F.1.3.
\end{classcode}
\end{abstract}

\section{Introduction}
\label{sect-intro}

In this paper, we develop theory about non-uniform complexity in a
setting in which the notion of single-pass instruction sequence
considered in program algebra is the central notion.

In the first place, we define a counterpart of the classical non-uniform
complexity class \PTpoly\ and formulate a counterpart of a well-known
complexity theoretic conjecture.
The conjecture in question is the conjecture that
$\NPT \not\subseteq \PTpoly$.
Some evidence for this conjecture is the Karp-Lipton
theorem~\cite{KL80a}, which says that the polynomial time hierarchy
collapses to the second level if $\NPT \subseteq \PTpoly$.
If the conjecture is right, then the conjecture that $\PT \neq \NPT$ is
right as well.
The counterpart of the former conjecture introduced in this paper is
called the non-uniform super-polynomial complexity hypothesis.
It is called a hypothesis instead of a conjecture because it is
primarily interesting for its consequences.

Over and above that, we define a counterpart of the non-uniform
complexity class \NPTpoly, introduce a notion of completeness for this
complexity class using a non-uniform reducibility relation, and
formulate three complexity hypotheses which concern restrictions on the
instruction sequences used for computation.
These three hypotheses are called super-polynomial feature elimination
complexity hypotheses.
The first of them is equivalent to the hypothesis that
$\NPTpoly \not\subseteq \PTpoly$.
We do not know whether there are equivalent hypotheses for the other two
hypotheses in well-known settings such as Turing machines with advice
and Boolean circuits.
All three hypotheses are intuitively appealing in the setting of
single-pass instruction sequences.

We show among other things that \PTpoly\ and \NPTpoly\ coincide with
their counterparts in the setting of single-pass instruction sequences
as defined in this paper and that a problem closely related to \iiiSAT\
is \NPT-complete as well as complete for the counterpart of \NPTpoly.

The work presented in this paper is part of a research program which is
concerned with different subjects from the theory of computation and the
area of computer architectures where we come across the relevancy of the
notion of instruction sequence.
The working hypothesis of this research program is that this notion is a
central notion of computer science.
It is clear that instruction sequence is a key concept in practice, but
strangely enough it has as yet not come prominently into the picture in
theoretical circles.

As part of this research program, issues concerning the following
subjects from the theory of computation have been investigated from the
viewpoint that a program is an instruction sequence: semantics of
programming languages~\cite{BM07e,BM07f},\linebreak[3]
expressiveness of programming languages~\cite{BM08h,BM07g}, and
computability~\cite{BM09m,BP04a}.
Performance related matters of instruction sequences have also been
investigated in the spirit of the theory of
computation~\cite{BM09i,BM07f}.
In the area of computer architectures, basic techniques aimed at
increasing processor performance have been studied as part of this
research program (see e.g.~\cite{BM06d}).

The above-mentioned work provides evidence for our hypothesis that the
notion of instruction sequence is a central notion of computer science.
To say the least, it shows that instruction sequences are relevant to
diverse subjects.
In addition, it is to be expected that the emerging developments with
respect to techniques for high-performance program execution on
classical or non-classical computers require that programs are
considered at the level of instruction sequences.
All this has motivated us to continue the above-mentioned research
program with the work on computational complexity presented in this
paper.

Program algebra~\cite{BL02a}, which is intended as a setting suited for
developing theory from the above-mentioned working hypothesis, is taken
for the basis of the development of theory under the research program.
Program algebra is not intended to provide a notation for programs that
is suited for actual programming.
With program algebra we have in view contemplation on programs rather
than construction of programs.

The starting-point of program algebra is the perception of a program as
a single-pass instruction sequence, i.e.\ a finite or infinite sequence
of instructions of which each instruction is executed at most once and
can be dropped after it has been executed or jumped over.
This perception is simple, appealing, and links up with practice.
The concepts underlying the primitives of program algebra are common in
programming, but the particular form of the primitives is not common.
The predominant concern in the design of program algebra has been to
achieve simple syntax and semantics, while maintaining the expressive
power of arbitrary finite control.
We do not need the whole of program algebra in this paper, but
nevertheless we will present the whole.
The purpose of that is to keep the grounds of its design recognizable.

A single-pass instruction sequence under execution is considered to
produce a behaviour to be controlled by some execution environment.
Threads as considered in basic thread algebra~\cite{BL02a} model such
behaviours:
upon each action performed by a thread, a reply from the execution
environment determines how the thread proceeds.
A thread may make use of services, i.e.\ components of the execution
environment.
Once introduced into threads and services, it is rather obvious that
each Turing machine can be simulated by means of a thread that makes use
of a service.
The thread and service correspond to the finite control and tape of the
Turing machine.
Simulation by means of a thread that makes use of a service is also
possible for other machines that have been proposed as a computational
model, such as register machines or multi-stack machines.

The threads that correspond to the finite controls of Turing machines
are examples of regular threads, i.e.\ threads that can only be in a
finite number of states.
The behaviours of all instruction sequences considered in program
algebra are regular threads and each regular thread is produced by some
instruction sequence.
This implies, for instance, that program algebra can be used to program
the finite control of any Turing machine.

In our study of non-uniform computational complexity, we are concerned
with functions that can be computed by finite instruction sequences
whose behaviours make use of services that make up Boolean registers.
The instruction sequences considered in program algebra are sufficient
to define a counterpart of \PTpoly, but not to define a counterpart of
\NPTpoly.
For a counterpart of \NPTpoly, we introduce an extension of program
algebra that allows for single-pass instruction sequences to split and
an extension of basic thread algebra with a behavioural counterpart of
instruction sequence splitting that is reminiscent of thread forking.

The approach to complexity followed in this paper is not suited to
uniform complexity.
This is not considered a great drawback.
Non-uniform complexity is the relevant notion of complexity when
studying what looks to be the major complexity issue in practice: the
scale-dependence of what is an efficient solution for a computational
problem.

This paper is organized as follows.
First, we review basic thread algebra and program algebra
(Sections~\ref{sect-BTA} and~\ref{sect-PGA}).
Next, we present mechanisms for interaction of threads with services and
give a description of Boolean register services
(Sections~\ref{sect-TSI} and~\ref{sect-Boolean-register}).
Then, we introduce the complexity class corresponding to \PTpoly\ and
formulate the non-uniform super-polynomial complexity hypothesis
(Sections~\ref{sect-PLIS} and~\ref{sect-hypothesis}).
After that, we present extensions of program algebra and basic thread
algebra needed in the subsequent sections
(Section~\ref{sect-splitting}).
Following this, we introduce the complexity class corresponding to
\NPTpoly\ and formulate the super-polynomial feature elimination
complexity hypotheses
(Sections~\ref{sect-PLSIS} and~\ref{sect-feature-elim}).
Finally, we make some concluding remarks (Section~\ref{sect-concl}).

Some familiarity with complexity theory is assumed.
The definitions of the complexity theoretic notions that are assumed
known can be found in many textbooks on computational complexity.
We mention~\cite{AB09a,BDG88a,Gol08a} as examples of textbooks in which
all the notions in question are introduced.

\section{Basic Thread Algebra}
\label{sect-BTA}

In this section, we review \BTA\ (Basic Thread Algebra), a form of
process algebra which is tailored to the description and analysis of the
behaviours of sequential programs under execution.
The behaviours concerned are called \emph{threads}.

In \BTA, it is assumed that a fixed but arbitrary set $\BAct$ of
\emph{basic actions}, with $\Tau \not\in \BAct$, has been given.
We write $\BActTau$ for $\BAct \union \set{\Tau}$.
The members of $\BActTau$ are referred to as \emph{actions}.

Threads proceed by performing actions in a sequential fashion.
Each basic action performed by a thread is taken as a command to be
processed by some service provided by the execution environment of the
thread.
The processing of a command may involve a change of state of the service
concerned.
At completion of the processing of the command, the service produces a
reply value.
This reply is either $\True$ or $\False$ and is returned to the thread
concerned.
Performing the action $\Tau$ will never lead to a state change and
always lead to the reply $\True$, but notwithstanding that its presence
matters.%
\footnote
{The action $\Tau$ reminds of the action $\tau$ used in process algebra.
 The Greek letter is not used here because the characteristic equations
 of the latter action are not~implied.}

\BTA\ has one sort: the sort $\Thr$ of \emph{threads}.
To build terms of sort $\Thr$, \BTA\ has the following constants and
operators:
\begin{iteml}
\item
the \emph{deadlock} constant $\const{\DeadEnd}{\Thr}$;
\item
the \emph{termination} constant $\const{\Stop}{\Thr}$;
\item
for each $a \in \BActTau$, the binary \emph{postconditional composition}
operator $\funct{\pcc{\ph}{a}{\ph}}{\linebreak[2]\Thr \x \Thr}{\Thr}$.
\end{iteml}
Terms of sort $\Thr$ are built as usual (see e.g.~\cite{ST99a,Wir90a}).
Throughout the paper, we assume that there is a countably infinite set
of variables of sort $\Thr$ which includes $x,y,z$.

We use infix notation for postconditional composition.
We introduce \emph{action prefixing} as an abbreviation: $a \bapf p$,
where $p$ is a term of sort $\Thr$, abbreviates $\pcc{p}{a}{p}$.

Let $p$ and $q$ be closed terms of sort $\Thr$ and $a \in \BActTau$.
Then $\pcc{p}{a}{q}$ will perform action $a$, and after that proceed as
$p$ if the processing of $a$ leads to the reply $\True$ (called a
positive reply), and proceed as $q$ if the processing of $a$ leads to
the reply $\False$ (called a negative reply).

\BTA\ has only one axiom.
This axiom is given in Table~\ref{axioms-BTA}.%
\begin{table}[!tb]
\caption{Axiom of \BTA}
\label{axioms-BTA}
\begin{eqntbl}
\begin{axcol}
\pcc{x}{\Tau}{y} = \pcc{x}{\Tau}{x}                    & \axiom{T1}
\end{axcol}
\end{eqntbl}
\end{table}
Using the abbreviation introduced above, axiom T1 can be written as
follows: $\pcc{x}{\Tau}{y} = \Tau \bapf x$.

Each closed \BTA\ term of sort $\Thr$ denotes a finite thread, i.e.\ a
thread of which the length of the sequences of actions that it can
perform is bounded.
Infinite threads can be defined by means of a set of recursion equations
(see e.g.~\cite{BM06a,BM07a}).
Regular threads, i.e.\ threads that can only be in a finite number of
states, can be defined by means of a finite set of recursion equations.

\section{Program Algebra}
\label{sect-PGA}

In this section, we review \PGA\ (ProGram Algebra).
The starting-point of \PGA\ is the perception of a program as a
single-pass instruction sequence, i.e.\ a finite or infinite sequence of
instructions of which each instruction is executed at most once and can
be dropped after it has been executed or jumped over.

In \PGA, it is assumed that there is a fixed but arbitrary set $\BInstr$
of \emph{basic instructions}.
\PGA\ has the following \emph{primitive instructions}:
\begin{iteml}
\item
for each $a \in \BInstr$, a \emph{plain basic instruction} $a$;
\item
for each $a \in \BInstr$, a \emph{positive test instruction} $\ptst{a}$;
\item
for each $a \in \BInstr$, a \emph{negative test instruction} $\ntst{a}$;
\item
for each $l \in \Nat$, a \emph{forward jump instruction} $\fjmp{l}$;
\item
a \emph{termination instruction} $\halt$.
\end{iteml}
We write $\PInstr$ for the set of all primitive instructions.

The intuition is that the execution of a basic instruction $a$ may
modify a state and produces $\True$ or $\False$ at its completion.
In the case of a positive test instruction $\ptst{a}$, basic instruction
$a$ is executed and execution proceeds with the next primitive
instruction if $\True$ is produced and otherwise the next primitive
instruction is skipped and execution proceeds with the primitive
instruction following the skipped one.
In the case where $\True$ is produced and there is not at least one
subsequent primitive instruction and in the case where $\False$ is
produced and there are not at least two subsequent primitive
instructions, deadlock occurs.
In the case of a negative test instruction $\ntst{a}$, the role of the
value produced is reversed.
In the case of a plain basic instruction $a$, the value produced is
disregarded: execution always proceeds as if $\True$ is produced.
The effect of a forward jump instruction $\fjmp{l}$ is that execution
proceeds with the $l$-th next instruction of the instruction sequence
concerned.
If $l$ equals $0$ or the $l$-th next instruction does not exist, then
$\fjmp{l}$ results in deadlock.
The effect of the termination instruction $\halt$ is that execution
terminates.

\PGA\ has the following constants and operators:
\begin{iteml}
\item
for each $u \in \PInstr$, an \emph{instruction} constant $u$\,;
\item
the binary \emph{concatenation} operator $\ph \conc \ph$\,;
\item
the unary \emph{repetition} operator $\ph\rep$\,.
\end{iteml}
Terms are built as usual.
Throughout the paper, we assume that there is a countably infinite set
of variables which includes $x,y,z$.

We use infix notation for concatenation and postfix notation for
repetition.

A closed \PGA\ term is considered to denote a non-empty, finite or
eventually periodic infinite sequence of primitive instructions.%
\footnote
{An eventually periodic infinite sequence is an infinite sequence with
 only finitely many distinct suffixes.}
Closed \PGA\ terms are considered equal if they represent the same
instruction sequence.
The axioms for instruction sequence equivalence are given in
Table~\ref{axioms-PGA}.%
\begin{table}[!tb]
\caption{Axioms of \PGA}
\label{axioms-PGA}
\begin{eqntbl}
\begin{axcol}
(x \conc y) \conc z = x \conc (y \conc z)              & \axiom{PGA1} \\
(x^n)\rep = x\rep                                      & \axiom{PGA2} \\
x\rep \conc y = x\rep                                  & \axiom{PGA3} \\
(x \conc y)\rep = x \conc (y \conc x)\rep              & \axiom{PGA4}
\end{axcol}
\end{eqntbl}
\end{table}
In this table, $n$ stands for an arbitrary natural number greater than
$0$.
For each $n > 0$, the term $x^n$ is defined by induction on $n$ as
follows: $x^1 = x$ and $x^{n+1} = x \conc x^n$.
The \emph{unfolding} equation $x\rep = x \conc x\rep$ is
derivable.
Each closed \PGA\ term is derivably equal to a term in
\emph{canonical form}, i.e.\ a term of the form $P$ or $P \conc Q\rep$,
where $P$ and $Q$ are closed \PGA\ terms in which the repetition
operator does not occur.

The members of the domain of an initial model of \PGA\ are called
\emph{instruction sequences}.
This is justified by the fact that one of the initial models of \PGA\ is
the model in which:
\begin{iteml}
\item
the domain is the set of all finite and eventually periodic infinite
sequences over the set $\PInstr$ of primitive instructions;
\item
the operation associated with ${} \conc {}$ is concatenation;
\item
the operation associated with ${}\rep$ is the operation ${}\srep$
defined as follows:
\begin{iteml}
\item
if $X$ is a finite sequence, then $X\srep$ is the unique eventually
periodic infinite sequence $Y$ such that $X$ concatenated $n$ times with
itself is a proper prefix of $Y$ for each $n \in \Nat$;
\item
if $X$ is an eventually periodic infinite sequence, then $X\srep$ is
$X$.
\end{iteml}
\end{iteml}
To simplify matters, we confine ourselves to this initial model of \PGA\
for the interpretation of \PGA\ terms.

The behaviours of the instruction sequences denoted by closed \PGA\
terms are considered to be regular threads, with the basic instructions
taken for basic actions.
Moreover, all regular threads in which $\Tau$ is absent are behaviours
of instruction sequences that can be denoted by closed \PGA\ terms
(see~\cite{PZ06a}, Proposition~2).
Closed \PGA\ terms in which the repetition operator does not occur
correspond to finite threads.

In the remainder of this paper, we consider instruction sequences that
can be denoted by closed \PGA\ terms in which the repetition operator
does not occur.
The \emph{thread extraction} operation $\extr{\ph}$ defined by the
equations given in Table~\ref{axioms-thread-extr} (for $a \in \BInstr$,
$l \in \Nat$ and $u \in \PInstr$) gives, for each closed \PGA\ term $P$
in which the repetition operator does not occur, a closed \BTA\ term
that denotes the behaviour of the instruction sequence denoted by $P$.%
\begin{table}[!tb]
\caption{Defining equations for thread extraction operation}
\label{axioms-thread-extr}
\begin{eqntbl}
\begin{eqncol}
\extr{a} = a \bapf \DeadEnd \\
\extr{a \conc x} = a \bapf \extr{x} \\
\extr{\ptst{a}} = a \bapf \DeadEnd \\
\extr{\ptst{a} \conc x} =
\pcc{\extr{x}}{a}{\extr{\fjmp{2} \conc x}} \\
\extr{\ntst{a}} = a \bapf \DeadEnd \\
\extr{\ntst{a} \conc x} =
\pcc{\extr{\fjmp{2} \conc x}}{a}{\extr{x}}
\end{eqncol}
\qquad
\begin{eqncol}
\extr{\fjmp{l}} = \DeadEnd \\
\extr{\fjmp{0} \conc x} = \DeadEnd \\
\extr{\fjmp{1} \conc x} = \extr{x} \\
\extr{\fjmp{l+2} \conc u} = \DeadEnd \\
\extr{\fjmp{l+2} \conc u \conc x} = \extr{\fjmp{l+1} \conc x} \\
\extr{\halt} = \Stop \\
\extr{\halt \conc x} = \Stop
\end{eqncol}
\end{eqntbl}
\end{table}

Henceforth, we will write \PGAfin\ for \PGA\ without the repetition
operator and axioms PGA2--PGA4, and we will write $\FIS$ for the set of
all instruction sequences that can be denoted by closed \PGAfin\ terms.
Moreover, we will write $\psize(X)$, where $X \in \FIS$, for the length
of $X$.

In addition to instruction sequence congruence, two coarser congruences
are introduced in~\cite{BL02a}.
We give additional axioms for those congruences, relating to the
instruction sequences that are considered in our study of non-uniform
computational complexity, in Appendix~\ref{app-beyond-is-congr}.

The use of a closed \PGAfin\ term is sometimes preferable to the use of
the corresponding closed \BTA\ term because thread extraction can give
rise to a combinatorial explosion.
For instance, suppose that $p$ is a closed \BTA\ term such that
\begin{ldispl}
p =
\extr{\overbrace{\ptst{a} \conc \ptst{b} \conc \ldots \conc
                 \ptst{a} \conc \ptst{b}}^{k\; \mathrm{times}} {} \conc
                 c \conc \halt}\;.
\end{ldispl}%
Then the size of $p$ is greater than $2^{k/2}$.
In Appendix~\ref{app-expl-subst}, we show that such combinatorial
explosions can be eliminated if we add explicit substitution to thread
algebra.

\section{Interaction of Threads with Services}
\label{sect-TSI}

A thread may make use of services.
That is, a thread may perform an action for the purpose of interacting
with a service that takes the action as a command to be processed.
The processing of an action may involve a change of state of the service
and at completion of the processing of the action the service returns a
reply value to the thread.
In this section, we introduce the use mechanism and the apply mechanism,
which are concerned with this kind of interaction between threads and
services.
The difference between the use mechanism and the apply mechanism is a
matter of perspective: the former is concerned with the effect of
services on threads and therefore produces threads, whereas the latter
is concerned with the effect of threads on services and therefore
produces services.

It is assumed that a fixed but arbitrary set $\Foci$ of \emph{foci} and
a fixed but arbitrary set $\Meth$ of \emph{methods} have been given.
Each focus plays the role of a name of some service provided by an
execution environment that can be requested to process a command.
Each method plays the role of a command proper.
For the set $\BAct$ of actions, we take the set
$\set{f.m \where f \in \Foci, m \in \Meth}$.
Performing an action $f.m$ is taken as making a request to the
service named $f$ to process command~$m$.

A \emph{service} $H$ consists of
\begin{iteml}
\item
a set $S$ of \emph{states};
\item
an \emph{effect} function $\funct{\eff}{\Meth \x S}{S}$;
\item
a \emph{yield} function
$\funct{\yld}{\Meth \x S}{\set{\True,\False,\Blocked}}$;
\item
an \emph{initial state} $s_0 \in S$;
\end{iteml}
satisfying the following condition:
\begin{ldispl}
\Forall{m \in \Meth, s \in S}
{(\yld(m,s) = \Blocked \Implies
  \Forall{m' \in \Meth}{\yld(m',\eff(m,s)) = \Blocked})}\;.
\end{ldispl}%
The set $S$ contains the states in which the service may be, and the
functions $\eff$ and $\yld$ give, for each method $m$ and state $s$, the
state and reply, respectively, that result from processing $m$ in state
$s$.

Given a service $H = \tup{S,\eff,\yld,s_0}$ and a method $m \in \Meth$:
\begin{iteml}
\item
the \emph{derived service} of $H$ after processing $m$, written
$\derive{m}H$, is the service $\tup{S,\eff,\yld,\eff(m,s_0)}$;
\item
the \emph{reply} of $H$ after processing $m$, written $H(m)$, is
$\yld(m,s_0)$.
\end{iteml}

A service $H$ can be understood as follows:
\begin{iteml}
\item
if a thread makes a request to the service to process $m$ and
$H(m) \neq \Blocked$, then the request is accepted, the reply is $H(m)$,
and the service proceeds as $\derive{m}H$;
\item
if a thread makes a request to the service to process $m$ and
$H(m) = \Blocked$, then the request is rejected and the service proceeds
as a service that rejects any request.
\end{iteml}
A service $H$ is called \emph{divergent} if $H(m) = \Blocked$ for all
$m \in \Meth$.
The effect of different divergent services on a thread is the same.
Therefore, all divergent services are identified.

We introduce the additional sort $\Serv$ of \emph{services} and the
following additional constant and operators:
\begin{iteml}
\item
the \emph{divergent service} constant $\const{\DivServ}{\Serv}$;
\item
for each $f \in \Foci$, the binary \emph{use} operator
$\funct{\use{\ph}{f}{\ph}}{\Thr \x \Serv}{\Thr}$;
\item
for each $f \in \Foci$, the binary \emph{apply} operator
$\funct{\apply{\ph}{f}{\ph}}{\Thr \x \Serv}{\Serv}$.
\end{iteml}
We use infix notation for the use and apply operators.

$\DivServ$ is a fixed but arbitrary divergent service.
The operators $\use{\ph}{f}{\ph}$ and $\apply{\ph}{f}{\ph}$ are
complementary.
Intuitively, $\use{p}{f}{H}$ is the thread that results from processing
all actions performed by thread $p$ that are of the form $f.m$ by
service $H$.
When an action of the form $f.m$ performed by thread $p$ is processed by
service $H$, that action is turned into the internal action $\Tau$ and
postconditional composition is removed in favour of action prefixing on
the basis of the reply value produced.
Intuitively, $\apply{p}{f}{H}$ is the service that results from
processing all basic actions performed by thread $p$ that are of the
form $f.m$ by service $H$.
When an action of the form $f.m$ performed by thread $p$ is processed by
service $H$, that service is changed into $\derive{m}H$.

The axioms for the use and apply operators are given in
Tables~\ref{axioms-use} and~\ref{axioms-apply}.%
\begin{table}[!tb]
\caption{Axioms for use operators}
\label{axioms-use}
\begin{eqntbl}
\begin{saxcol}
\use{\Stop}{f}{H} = \Stop                            & & \axiom{TSU1} \\
\use{\DeadEnd}{f}{H} = \DeadEnd                      & & \axiom{TSU2} \\
\use{(\Tau \bapf x)}{f}{H} =
                          \Tau \bapf (\use{x}{f}{H}) & & \axiom{TSU3} \\
\use{(\pcc{x}{g.m}{y})}{f}{H} =
\pcc{(\use{x}{f}{H})}{g.m}{(\use{y}{f}{H})}
 & \mif f \neq g                                       & \axiom{TSU4} \\
\use{(\pcc{x}{f.m}{y})}{f}{H} =
\Tau \bapf (\use{x}{f}{\derive{m}H})
                                & \mif H(m) = \True    & \axiom{TSU5} \\
\use{(\pcc{x}{f.m}{y})}{f}{H} =
\Tau \bapf (\use{y}{f}{\derive{m}H})
                                & \mif H(m) = \False   & \axiom{TSU6} \\
\use{(\pcc{x}{f.m}{y})}{f}{H} = \DeadEnd
                                & \mif H(m) = \Blocked & \axiom{TSU7}  \\
\use{(\pcc{x}{f.m}{y})}{f}{\DivServ} = \DeadEnd      & & \axiom{TSU8}
\end{saxcol}
\end{eqntbl}
\end{table}
\begin{table}[!tb]
\caption{Axioms for apply operators}
\label{axioms-apply}
\begin{eqntbl}
\begin{saxcol}
\apply{\Stop}{f}{H} = H                              & & \axiom{TSA1} \\
\apply{\DeadEnd}{f}{H} = \DivServ                    & & \axiom{TSA2} \\
\apply{(\Tau \bapf x)}{f}{H} = \apply{x}{f}{H}       & & \axiom{TSA3} \\
\apply{(\pcc{x}{g.m}{y})}{f}{H} = \DivServ
                                       & \mif f \neq g & \axiom{TSA4} \\
\apply{(\pcc{x}{f.m}{y})}{f}{H} = \apply{x}{f}{\derive{m}H}
                                & \mif H(m) = \True    & \axiom{TSA5} \\
\apply{(\pcc{x}{f.m}{y})}{f}{H} = \apply{y}{f}{\derive{m}H}
                                & \mif H(m) = \False   & \axiom{TSA6} \\
\apply{(\pcc{x}{f.m}{y})}{f}{H} = \DivServ
                                & \mif H(m) = \Blocked & \axiom{TSA7} \\
\apply{(\pcc{x}{f.m}{y})}{f}{\DivServ} = \DivServ    & & \axiom{TSA8}
\end{saxcol}
\end{eqntbl}
\end{table}
In these tables, $f$ and $g$ stand for arbitrary foci from $\Foci$, $m$
stands for an arbitrary method from $\Meth$, and $H$ is a variable of
sort $\Serv$.
Axioms TSU3 and TSU4 express that the action $\Tau$ and actions of
the form $g.m$, where $f \neq g$, are not processed.
Axioms TSU5 and TSU6 express that a thread is affected by a service
as described above when an action of the form $f.m$ performed by the
thread is processed by the service.
Axiom TSU7 expresses that deadlock takes place when an action to be
processed is not accepted.
Axiom TSU8 expresses that the divergent service does not accept any
action.
Axiom TSA3 expresses that a service is not affected by a thread when the
action $\Tau$ is performed by the thread and axiom TSA4 expresses that a
service is turned into the divergent service when an action of the form
$g.m$, where $f \neq g$,  is performed by the thread.
Axioms TSA5 and TSA6 express that a service is affected by a thread as
described above when an action of the form $f.m$ performed by the thread
is processed by the service.
Axiom TSA7 expresses that a service is turned into the divergent service
when an action performed by the thread is not accepted.
Axiom TSA8 expresses that the divergent service is not affected by a
thread when an action of the form $f.m$ is performed by the thread.

\section{Instruction Sequences Acting on Boolean Registers}
\label{sect-Boolean-register}

Our study of computational complexity is concerned with instruction
sequences that act on Boolean registers.
In this section, we describe services that make up Boolean registers.
We also introduce special foci that serve as names of Boolean register
services.

The Boolean register services accept the following methods:
\begin{itemize}
\item
a \emph{set to true method} $\setbr{\True}$;
\item
a \emph{set to false method} $\setbr{\False}$;
\item
a \emph{get method} $\getbr$.
\end{itemize}
We write $\Methbr$ for the set
$\set{\setbr{\True},\setbr{\False},\getbr}$.
It is assumed that $\Methbr \subseteq \Meth$.

The methods accepted by Boolean register services can be explained as
follows:
\begin{itemize}
\item
$\setbr{\True}$\,:
the contents of the Boolean register becomes $\True$ and the reply is
$\True$;
\item
$\setbr{\False}$\,:
the contents of the Boolean register becomes $\False$ and the reply is
$\False$;
\item
$\getbr$\,:
nothing changes and the reply is the contents of the Boolean register.
\end{itemize}

Let $s \in \set{\True,\False,\Blocked}$.
Then the \emph{Boolean register service} with initial state $s$, written
$\BR_s$, is the service $\tup{\set{\True,\False,\Blocked},\eff,\eff,s}$,
where the function $\eff$ is defined as follows
($b \in \set{\True,\False}$):
\begin{ldispl}
\begin{geqns}
\eff(\setbr{\True},b) = \True\;,\;
\\
\eff(\setbr{\False},b) = \False\;,
\\
\eff(\getbr,b) = b\;,
\end{geqns}
\qquad\qquad
\begin{geqns}
\eff(m,b) = \Blocked & \mif m \not\in \Methbr\;,
\\
\eff(m,\Blocked) = \Blocked\;.
\end{geqns}
\end{ldispl}%
Notice that the effect and yield functions of a Boolean register service
are the same.
This means that at completion of the processing of a method the state
that results from the processing is returned as the reply.

In the instruction sequences which concern us in the remainder of this
paper, a number of Boolean registers is used as input registers, a
number of Boolean registers is used as auxiliary registers, and one
Boolean register is used as output register.

It is assumed that $\inbr{1},\inbr{2},\ldots \in \Foci$,
$\auxbr{1},\auxbr{2},\ldots \in \Foci$, and $\outbr \in \Foci$.
These foci play special roles:
\begin{iteml}
\item
for each $i \in \Natpos$,%
\footnote
{We write $\Natpos$ for the set $\set{n \in \Nat \where n > 0}$.}
$\inbr{i}$ serves as the name of the Boolean register that is used as
$i$-th input register in instruction sequences;
\item
for each $i \in \Natpos$, $\auxbr{i}$ serves as the name of the Boolean
register that is used as $i$-th auxiliary register in instruction
sequences;
\item
$\outbr$ serves as the name of the Boolean register that is used as
output register in instruction sequences.
\end{iteml}

Henceforth, we will write
$\Fociin$ for $\set{\inbr{i} \where i \in \Natpos}$ and
$\Fociaux$ for $\set{\auxbr{i} \where i \in \Natpos}$.
Moreover, we will write $\ISbr$ for the set of all instruction sequences
from $\FIS$ in which all primitive instructions, with the exception of
jump instructions and the termination instruction, contain only basic
instructions from the set
\begin{ldispl}
\set{f.\getbr \where f \in \Fociin \union \Fociaux} \union
\set{f.\setbr{b} \where f \in \Fociaux \union \set{\outbr} \And
                        b \in \set{\True,\False}}
\end{ldispl}%
and $\ISbrna$ for the set of all instruction sequences from $\FIS$ in
which all primitive instructions, with the exception of jump
instructions and the termination instruction, contain only basic
instructions from the set
\begin{ldispl}
\set{f.\getbr \where f \in \Fociin} \union
\set{\outbr.\setbr{b} \where b \in \set{\True,\False}}\;.
\end{ldispl}%
$\ISbrna$ is the set of all instruction sequences from $\ISbr$ in which
no auxiliary registers are used.
$\ISbr$ is the set of all instruction sequences from $\FIS$ that matter
to the complexity class \PLIS\ which will be introduced in
Section~\ref{sect-PLIS}.

\section{The Complexity Class \PLIS}
\label{sect-PLIS}

In the field of computational complexity, it is quite common to study
the complexity of computing functions on finite strings over a binary
alphabet.
Since strings over an alphabet of any fixed size can be efficiently
encoded as strings over a binary alphabet, it is sufficient to consider
only a binary alphabet.
We adopt the set $\Bool = \set{\True,\False}$ as preferred binary
alphabet.

An important special case of functions on finite strings over a binary
alphabet is the case where the value of functions is restricted to
strings of length $1$.
Such a function is often identified with the set of strings of which it
is the characteristic function.
The set in question is usually called a language or a decision problem.
The identification mentioned above allows of looking at the problem of
computing a function $\funct{f}{\seqof{\Bool}}{\Bool}$ as the problem of
deciding membership of the set
$\set{w \in \seqof{\Bool} \where f(w) = \True}$.

With each function $\funct{f}{\seqof{\Bool}}{\Bool}$, we can associate
an infinite sequence $\indfam{f_n}{n \in \Nat}$ of functions, with
$\funct{f_n}{\Bool^n}{\Bool}$ for every $n \in \Nat$, such that $f_n$ is
the restriction of $f$ to $\Bool^n$  for each $n \in \Nat$.
The complexity of computing such sequences of functions, which we call
Boolean function families, is studied in the remainder of this paper.
In the current section, we introduce the class \PLIS\ of all Boolean
function families that can be computed by polynomial-length instruction
sequences from $\ISbr$.

An \emph{$n$-ary Boolean function} is a function
$\funct{f}{\Bool^n}{\Bool}$.
Let $\phi$ be a Boolean formula containing the variables
$v_1,\ldots,v_n$.
Then $\phi$ induces an $n$-ary Boolean function $f_n$ such that
$f_n(b_1,\ldots,b_n) = \True$ iff $\phi$ is satisfied by the assignment
$\sigma$ to the variables $v_1,\ldots,v_n$ defined by
$\sigma(v_1) = b_1$, \ldots, $\sigma(v_n) = b_n$.
The Boolean function in question is called the Boolean function
\emph{induced} by $\phi$.

A \emph{Boolean function family} is an infinite sequence
$\indfam{f_n}{n \in \Nat}$ of functions, where $f_n$ is an $n$-ary
Boolean function for each $n \in \Nat$.
A Boolean function family $\indfam{f_n}{n \in \Nat}$ can be identified
with the unique function $\funct{f}{\seqof{\Bool}}{\Bool}$ such that for
each $n \in \Nat$, for each $w \in \Bool^n$, $f(w) = f_n(w)$.
In this paper, we are concerned with non-uniform complexity.
Considering sets of Boolean function families as complexity classes
looks to be most natural when studying non-uniform complexity.
We will make the identification mentioned above only where connections
with well-known complexity classes are made.

Let $n \in \Nat$, let $\funct{f}{\Bool^n}{\Bool}$, and let
$X \in \ISbr$.
Then $X$ \emph{computes} $f$ if there exists an $l \in \Nat$ such that
for all $b_1,\ldots,b_n \in \Bool$:
\begin{ldispl}
( \ldots
 (( \ldots
   (\extr{X}
     \useop{\auxbr{1}} \BR_\False) \ldots \useop{\auxbr{l}} \BR_\False)
   \useop{\inbr{1}} \BR_{b_1}) \ldots \useop{\inbr{n}} \BR_{b_n})
 \applyop{\outbr} \BR_\False
\\ \quad {} = \BR_{f(b_1,\ldots,b_n)}\;.
\end{ldispl}%

$\PLIS$ is the class of all Boolean function families
$\indfam{f_n}{n \in \Nat}$ that satisfy:
\begin{quote}
there exists a polynomial function $\funct{h}{\Nat}{\Nat}$ such that
for all $n \in \Nat$ there exists an $X \in \ISbr$ such that
$X$ computes $f_n$ and $\psize(X) \leq h(n)$.
\end{quote}

The question arises whether all $n$-ary Boolean functions can be
computed by an instruction sequence from $\ISbr$. This question can
answered in the affirmative.
\begin{theorem}
\label{theorem-comput-boolfunc}
For each $n \in \Nat$, for each $n$-ary Boolean function
$\funct{f_n}{\Bool^n}{\Bool}$, there exists an $X \in \ISbrna$ in which
no other jump instruction than $\fjmp{2}$ occurs such that $X$ computes
$f_n$ and $\psize(X) = O(n \mul 2^n)$.
\end{theorem}
\begin{proof}
The following is well-known (see e.g.~\cite{AB09a}, Claim~2.14): for
each $n$-ary Boolean function $\funct{f_n}{\Bool^n}{\Bool}$, there is a
\CNF-formula $\phi$ containing $n$ variables such that
$\funct{f_n}{\Bool^n}{\Bool}$ is the Boolean function induced by $\phi$
and the size of $\phi$ is $n \mul 2^n$.
Therefore, it is sufficient to show that, for each \CNF-formula $\phi$
containing the variables $v_1,\ldots,v_n$, there exists an
$X \in \ISbrna$ in which no other jump instruction than $\fjmp{2}$
occurs such that $X$ computes the Boolean function induced by $\phi$ and
$\psize(X)$ is linear in the size of $\phi$.

Let $\inseqcnf$ be the function from the set of all \CNF-formulas
containing the variables $v_1,\ldots,v_n$ to $\ISbrna$ as follows:
\begin{ldispl}
\inseqcnf(\AND{i \in [1,m]} \OR{j \in [1,n_i]} \xi_{ij}) = {}
\\ \quad
\inseqcnf'(\xi_{11}) \conc \ldots \conc \inseqcnf'(\xi_{1n_1}) \conc
\ptst{\outbr.\setbr{\False}} \conc \fjmp{2} \conc \halt \conc
\\ \qquad \vdots \\ \quad
\inseqcnf'(\xi_{m1}) \conc \ldots \conc \inseqcnf'(\xi_{mn_m}) \conc
\ptst{\outbr.\setbr{\False}} \conc \fjmp{2} \conc \halt \conc
\ptst{\outbr.\setbr{\True}} \conc \halt\;,
\end{ldispl}%
where
\begin{ldispl}
\begin{aeqns}
\inseqcnf'(v_k)     & = & \ptst{\inbr{k}.\getbr} \conc \fjmp{2}\;,
\\
\inseqcnf'(\Not v_k)& = & \ntst{\inbr{k}.\getbr} \conc \fjmp{2}\;.
\end{aeqns}
\end{ldispl}%
Recall that a disjunction is satisfied if one of its disjuncts is
satisfied and a conjunction is satisfied if each of its conjuncts is
satisfied.
Using these facts, it is easy to prove by induction on the number of
clauses in a \CNF-formula, and in the basis step by induction on the
number of literals in a clause, that no other jump instruction than
$\fjmp{2}$ occurs in $\inseqcnf(\phi)$ and that $\inseqcnf(\phi)$
computes the Boolean function induced by $\phi$.
Moreover, it is easy to see that $\psize(\inseqcnf(\phi))$ is linear in
the size of $\phi$.
\qed
\end{proof}

In the proof of Theorem~\ref{theorem-comput-boolfunc}, it is shown that
the Boolean function induced by a \CNF-formula can be computed, without
using auxiliary Boolean registers, by an instruction sequence from
$\ISbrna$ that contains no other jump instructions than $\fjmp{2}$ and
whose length is linear in the size of the formula.
If we permit arbitrary jump instructions, this result generalizes from
\CNF-formulas to arbitrary basic Boolean formulas, i.e.\ Boolean
formulas in which no other connectives than $\Not$, $\Or$ and $\And$
occur.
\begin{theorem}
\label{theorem-comput-boolform}
For each basic Boolean formula $\phi$, there exists an $X \in \ISbrna$
in which the basic instruction $\outbr.\setbr{\False}$ does not occur
such that $X$ computes the Boolean function induced by $\phi$ and
$\psize(X)$ is linear in the size of $\phi$.
\end{theorem}
\begin{proof}
Let $\inseq$ be the function from the set of all basic Boolean formulas
containing the variables $v_1,\ldots,v_n$ to $\ISbrna$ as follows:
\begin{ldispl}
\inseq(\phi) =
\inseq'(\phi) \conc \ptst{\outbr.\setbr{\True}} \conc \halt\;,
\end{ldispl}%
where
\begin{ldispl}
\begin{aeqns}
\inseq'(v_k) & = & \ptst{\inbr{k}.\getbr}\;,
\\
\inseq'(\Not \phi) & = & \inseq'(\phi) \conc \fjmp{2}\;,
\\
\inseq'(\phi \Or \psi) & = &
\inseq'(\phi) \conc
\fjmp{\psize(\inseq'(\psi)) + 1} \conc \inseq'(\psi)\;,
\\
\inseq'(\phi \And \psi) & = &
\inseq'(\phi) \conc \fjmp{2} \conc
\fjmp{\psize(\inseq'(\psi)) + 2} \conc \inseq'(\psi)\;.
\end{aeqns}
\end{ldispl}%
Using the same facts about disjunctions and conjunctions as in the proof
of Theorem~\ref{theorem-comput-boolfunc}, it is easy to prove by
induction on the structure of $\phi$ that $\inseq(\phi)$ computes the
Boolean function induced by $\phi$.
Moreover, it is easy to see that $\psize(\inseq(\phi))$ is linear in the
size of $\phi$.
\qed
\end{proof}

We consider the proof of Theorem~\ref{theorem-comput-boolfunc} once
again.
The instruction sequences yielded by the function $\inseqcnf$ contain
the jump instruction $\fjmp{2}$.
Each occurrence of $\fjmp{2}$ belongs to a jump chain ending in
the instruction sequence $\ptst{\outbr.\setbr{\True}} \conc \halt$.
Therefore, each occurrence of $\fjmp{2}$ can safely be replaced by the
instruction $\ptst{\outbr.\setbr{\False}}$, which like $\fjmp{2}$ skips
the next instruction.
Moreover, the occurrences of the instruction sequence
$\ptst{\outbr.\setbr{\False}} \conc \fjmp{2} \conc \halt$ can be
replaced by the instruction $\halt$ because the content of the Boolean
register concerned is initially $\False$.
The former point gives rise to the following interesting corollary.
\begin{corollary}
\label{corollary-comput-boolfunc}
For each $n \in \Nat$, for each $n$-ary Boolean function
$\funct{f_n}{\Bool^n}{\Bool}$, there exists an $X \in \ISbrna$ in which
jump instructions do not occur such that $X$ computes $f_n$ and
$\psize(X) = O(n \mul 2^n)$.
\end{corollary}

In Corollary~\ref{corollary-comput-boolfunc}, the instruction sequences
in question contain no jump instructions.
However, they contain multiple termination instructions and both
$\outbr.\setbr{\True}$ and $\outbr.\setbr{\False}$.
This raises the question whether further restrictions are possible.
We have a negative result.
\begin{theorem}
\label{theorem-comput-negative}
Let $\phi$ be the Boolean formula $v_1 \And v_2 \And v_3$.
Then there does not exist an $X \in \ISbrna$ in which jump instructions
do not occur, multiple termination instructions do not occur and the
basic instruction $\outbr.\setbr{\False}$ does not occur such that $X$
computes the Boolean function induced by $\phi$.
\end{theorem}
\begin{proof}
Suppose that $X = u_1 \conc \ldots \conc u_k$ is an instruction sequence
from $\ISbrna$ satisfying the restrictions and computing the Boolean
function induced by $\phi$.
Consider the smallest $l \in [1,k]$ such that $u_l$ is either
$\outbr.\setbr{\True}$, $\ptst{\outbr.\setbr{\True}}$ or
$\ntst{\outbr.\setbr{\True}}$ (there must be such an $l$).
Because $\phi$ is not satisfied by all assignments to the variables
$v_1,v_2,v_3$, it cannot be the case that $l = 1$.
In the case where $l > 1$, for each $i \in [1,l-1]$, $u_i$ is either
$\inbr{j}.\getbr$, $\ptst{\inbr{j}.\getbr}$ or $\ntst{\inbr{j}.\getbr}$
for some $j \in \set{1,2,3}$.
This implies that, for each $i \in [0,l-1]$, there exists a basic
Boolean formula $\psi_i$ over the variables $v_1,v_2,v_3$ that is unique
up to logical equivalence such that, for each $b_1,b_2,b_3 \in \Bool$,
if the initial states of the Boolean registers named $\inbr{1}$,
$\inbr{2}$ and $\inbr{3}$ are $b_1$, $b_2$ and $b_3$, respectively, then
$u_{i+1}$ will be executed iff $\psi_i$ is satisfied by the assignment
$\sigma$ to the variables $v_1,v_2,v_3$ defined by
$\sigma(v_1) = b_1$, $\sigma(v_2) = b_2$ and $\sigma(v_3) = b_2$.
We have that $\psi_0 \Iff \True$ and,
for each $i \in [1,l-1]$,
$\psi_i \Iff (\psi_{i-1} \Implies \True)$ if
 $u_i \equiv \inbr{j}.\getbr$,
$\psi_i \Iff (\psi_{i-1} \Implies v_j)$ if
 $u_i \equiv \ptst{\inbr{j}.\getbr}$, and
$\psi_i \Iff (\psi_{i-1} \Implies \Not v_j)$ if
 $u_i \equiv \ntst{\inbr{j}.\getbr}$.
Hence, for each $i \in [0,l-1]$, $\psi_i \Implies \phi$ implies
$\True \Implies \phi$ or $v_j \Implies \phi$ or $\Not v_j \Implies \phi$
for some $j \in \set{1,2,3}$.
Because the latter three Boolean formulas are no tautologies,
$\psi_i \Implies \phi$ is no tautology either.
This means that, for each $i \in [1,l-1]$,
$\psi_i \Implies \phi$ is not satisfied by all assignments to the
variables $v_1,v_2,v_3$.
Hence, $X$ cannot exist.
\qed
\end{proof}

Because the content of the Boolean register concerned is initially
$\False$, the question arises whether $\outbr.\setbr{\False}$ is
essential in instruction sequences computing Boolean functions.
This question can be answered in the affirmative if we permit the use of
auxiliary Boolean registers.
\begin{theorem}
\label{theorem-comput-output}
Let $n \in \Nat$, let $\funct{f}{\Bool^n}{\Bool}$, and let $X \in \ISbr$
be such that $X$ computes $f$.
Then there exists an $Y \in \ISbr$ in which the basic instruction
$\outbr.\setbr{\False}$ does not occur such that $Y$ computes $f$ and
$\psize(Y)$ is linear in $\psize(X)$.
\end{theorem}
\begin{proof}
Let $o \in \Natpos$ be such that the basic instructions
$\auxbr{o}.\setbr{\True}$, $\auxbr{o}.\setbr{\False}$, and
$\auxbr{o}.\getbr$ do not occur in $X$.
Let $X'$ be obtained from $X$ by replacing each occurrence of the focus
$\outbr$ by $\auxbr{o}$.
Suppose that $X' = u_1 \conc \ldots \conc u_k$.
Let $Y$ be obtained from $u_1 \conc \ldots \conc u_k$ as follows:
\begin{enumerate}
\item
stop if $u_1 \equiv \halt$;
\item
stop if there exists no $j \in [2,k]$ such that
$u_{j-1} \not\equiv \outbr.\setbr{\True}$ and $u_j \equiv \halt$;
\item
find the least $j \in [2,k]$ such that
$u_{j-1} \not\equiv \outbr.\setbr{\True}$ and $u_j \equiv \halt$;
\item
replace $u_j$ by
$\ptst{\auxbr{o}.\getbr} \conc \outbr.\setbr{\True} \conc \halt$,
\item
for each $i \in [1,k]$, replace $u_i$ by $\fjmp{l + 2}$ if
$u_i \equiv \fjmp{l}$ and $i < j < i + l$;
\item
repeat the preceding steps for the resulting instruction sequence.
\end{enumerate}
It is easy to prove by induction on $k$ that the Boolean function
computed by $X$ and Boolean function computed by $Y$ are the same.
Moreover, it is easy to see that $\psize(Y) < 3 \mul \psize(X)$.
Hence, $\psize(Y)$ is linear in $\psize(X)$.
\qed
\end{proof}

Because Boolean formulas can be looked upon as Boolean circuits in which
all gates have out-degree $1$, the question arises whether
Theorem~\ref{theorem-comput-boolform} generalizes from Boolean formulas
to Boolean circuits.
This question can be answered in the affirmative if we permit the use of
auxiliary Boolean registers.
\begin{theorem}
\label{theorem-comput-boolcirc}
For each Boolean circuit $C$ containing no other gates than
$\Not$-gates, $\Or$-gates and $\And$-gates, there exists an
$X \in \ISbr$ in which the basic instruction $\outbr.\setbr{\False}$
does not occur such that $X$ computes the Boolean function induced by
$C$ and $\psize(X)$ is linear in the size of $C$.
\end{theorem}
\begin{proof}
Let $\inseqc$ be the function from the set of all Boolean circuits with
input nodes $\inode{1},\ldots,\inode{n}$ and gates
$\gate{1},\ldots,\gate{m}$ to $\ISbrna$ as follows:
\begin{ldispl}
\inseqc(C) =
\inseqc'(\gate{1}) \conc \ldots \conc \inseqc'(\gate{m}) \conc
\ptst{\auxbr{m}.\getbr} \conc
\ptst{\outbr.\setbr{\True}} \conc \halt\;,
\end{ldispl}%
where
\begin{ldispl}
\inseqc'(\gate{k}) = {} \\ \quad
\inseqc''(\pnode) \conc \fjmp{2} \conc \ptst{\auxbr{k}.\setbr{\True}}
\\ \quad
\mbox
 {if $\gate{k}$ is a $\Not$-gate with direct preceding node $\pnode$}\;,
\\
\inseqc'(\gate{k}) = {} \\ \quad
\inseqc''(\pnode) \conc \fjmp{2} \conc
\inseqc''(\pnode') \conc \ptst{\auxbr{k}.\setbr{\True}}
\\ \quad
\mbox
 {if $\gate{k}$ is a $\Or$-gate with direct preceding nodes $\pnode$
  and $\pnode'$}\;,
\\
\inseqc'(\gate{k}) = {} \\ \quad
\inseqc''(\pnode) \conc \fjmp{2} \conc \fjmp{3} \conc
\inseqc''(\pnode') \conc \ptst{\auxbr{k}.\setbr{\True}}
\\ \quad
\mbox
 {if $\gate{k}$ is a $\And$-gate with direct preceding nodes $\pnode$
  and $\pnode'$}\;,
\end{ldispl}%
and
\begin{ldispl}
\begin{aceqns}
\inseqc''(\inode{k}) & = & \ptst{\inbr{k}.\getbr}\;,
\\
\inseqc''(\gate{k}) & = & \ptst{\auxbr{k}.\getbr}\;.
\end{aceqns}
\end{ldispl}%
Using the same facts about disjunctions and conjunctions as in the
proofs of Theorems~\ref{theorem-comput-boolfunc}
and~\ref{theorem-comput-boolform}, it is easy to prove by induction on
the depth of $C$ that $\inseqc(C)$ computes the Boolean function induced
by $C$ if $\gate{1},\ldots,\gate{m}$ is a topological sorting of the
gates of $C$.
Moreover, it is easy to see that $\psize(\inseqc(C))$ is linear in the
size of $C$.
\qed
\end{proof}

Henceforth, we write $\phi(b_1,\ldots,b_n)$, where $\phi$ is a Boolean
formula containing the variables $v_1,\ldots,v_n$ and
$b_1,\ldots,b_n \in \Bool$, to indicate that $\phi$ is satisfied by the
assignment $\sigma$ to the variables $v_1,\ldots,v_n$ defined by
$\sigma(v_1) = b_1$, \ldots, $\sigma(v_n) = b_n$.

\PLIS\ includes Boolean function families that correspond to
uncomputable functions from $\seqof{\Bool}$ to $\Bool$.
Take an undecidable set $N \subseteq \Nat$ and consider the Boolean
function family $\indfam{f_n}{n \in \Nat}$ with, for each $n \in \Nat$,
$\funct{f_n}{\Bool^n}{\Bool}$ defined by
\begin{ldispl}
\begin{gceqns}
f_n(b_1,\ldots,b_n) = \True  & \mif n \in N\;,
\\
f_n(b_1,\ldots,b_n) = \False & \mif n \not\in N\;.
\end{gceqns}
\end{ldispl}%
For each $n \in N$, $f_n$ is computed by the instruction sequence
$\outbr.\setbr{\True} \conc \halt$.
For each $n \not\in N$, $f_n$ is computed by the instruction sequence
$\outbr.\setbr{\False} \conc \halt$.
The length of these instruction sequences is constant in $n$.
Hence, $\indfam{f_n}{n \in \Nat}$ is in \PLIS.
However, the corresponding function $\funct{f}{\seqof{\Bool}}{\Bool}$ is
clearly uncomputable.
This reminds of the fact that \PTpoly\ includes uncomputable functions
from $\seqof{\Bool}$ to $\Bool$.

It happens that \PLIS\ and \PTpoly\ coincide, provided that we identify
each Boolean function family $\indfam{f_n}{n \in \Nat}$ with the unique
function $\funct{f}{\seqof{\Bool}}{\Bool}$ such that for each
$n \in \Nat$, for each $w \in \Bool^n$, $f(w) = f_n(w)$.
\begin{theorem}
\label{theorem-PLIS-is-PTpoly}
$\PLIS = \PTpoly$.
\end{theorem}
\begin{proof}
We will prove the inclusion $\PTpoly \subseteq \PLIS$ using the
definition of $\PTpoly$ in terms of Boolean circuits and we will prove
the inclusion $\PLIS \subseteq \PTpoly$ using the definition of
$\PTpoly$ in terms of Turing machines that take advice.

$\PTpoly \subseteq \PLIS$:
Suppose that $\indfam{f_n}{n \in \Nat}$ in $\PTpoly$.
Then, for all $n \in \Nat$, there exists a Boolean circuit $C$ such that
$C$ computes $f_n$ and the size of $C$ is polynomial in $n$.
For each $n \in \Nat$, let $C_n$ be such a $C$.
From Theorem~\ref{theorem-comput-boolcirc} and the fact that linear in
the size of $C_n$ implies polynomial in $n$, it follows that each
Boolean function family in $\PTpoly$ is also in $\PLIS$.

$\PLIS \subseteq \PTpoly$:
Suppose that $\indfam{f_n}{n \in \Nat}$ is in $\PLIS$.
Then, for all $n \in \Nat$, there exists an $X \in \ISbr$ such that $X$
computes $f_n$ and $\psize(X)$ is polynomial in $n$.
For each $n \in \Nat$, let $X_n$ be such an $X$.
Then $f$ can be computed by a Turing machine that, on an input of size
$n$, takes a binary description of $X_n$ as advice and then just
simulates the execution of $X_n$.
It is easy to see that, under the assumption that instructions
$\auxbr{i}.m$, $\ptst{\auxbr{i}.m}$, $\ntst{\auxbr{i}.m}$ and $\fjmp{i}$
with $i > \psize(X_n)$ do not occur in $X_n$, the size of the
description of $X_n$ and the number of steps that it takes to simulate
the execution of $X_n$ are both polynomial in $n$.
It is obvious that we can make the assumption without loss of
generality.
Hence, each Boolean function family in $\PLIS$ is also in $\PTpoly$.
\qed
\end{proof}

We do not know whether there are restrictions on the number of auxiliary
Boolean registers in the definition of \PLIS\ that lead to a class
different from \PLIS.
In particular, it is unknown to us whether the restriction to zero
auxiliary Boolean registers leads to a class different from \PLIS.

\section{The Non-uniform Super-polynomial Complexity Hypothesis}
\label{sect-hypothesis}

In this section, we introduce a complexity hypothesis which is a
counterpart of the classical complexity theoretic conjecture that
$\NPT \not\subseteq \PTpoly$ in the current setting.
The counterpart in question corresponds to the conjecture that
$\iiiSAT \notin \PTpoly$.
By the \NPT-completeness of \iiiSAT, these conjectures are equivalent.
If they are right, then the conjecture that $\NPT \neq \PT$ is right as
well.
We talk here about a hypothesis instead of a conjecture because we are
primarily interested in its consequences.

To formulate the hypothesis, we need a Boolean function family
$\indfam{\iiiSATC_n}{n \in \Nat}$ that corresponds to \iiiSAT.
We obtain this Boolean function family by encoding \iiiCNF-formulas as
sequences of Boolean values.

We write $\ndisj(k)$ for
$\binom{2k}{1} + \binom{2k}{2} + \binom{2k}{3}$.
$\ndisj(k)$ is the number of combinations of at most $3$
elements from a set with $2k$ elements.
Notice that $\ndisj(k) = (4k^3 + 5k) / 3$.

It is assumed that a countably infinite set $\set{v_1,v_2,\ldots}$ of
propositional variables has been given.
Moreover, it is assumed that a family of bijections
\begin{ldispl}
\indfam
 {\funct{\alpha_k}{\left[1,\ndisj(k)\right]}
        {\set{L \subseteq \set{v_1,\Not v_1,\ldots,v_k,\Not v_k} \where
              1 \leq \card(L) \leq 3}}}
 {k \in \Nat}
\end{ldispl}%
has been given that satisfies the following two conditions:
\begin{ldispl}
\Forall{i \in \Nat}
 {\Forall{j \in [1,\ndisj(i)]}{{\alpha_i}^{-1}(\alpha_{i+1}(j)) = j}}\;,
\\
\alpha\; \mbox{is polynomial-time computable}\;,
\end{ldispl}
\par
where
$\funct{\alpha}{\Natpos}
       {\set{L \subseteq \set{v_1,\Not v_1,v_2,\Not v_2,\ldots} \where
             1 \leq \card(L) \leq 3}}$
is defined by
\begin{ldispl}
\mbox{} \hspace*{\leftmargin}
\alpha(i) = \alpha_{\min \set{j \where i \in [1,\ndisj(j)]}}(i)\;.
\end{ldispl}%
The function $\alpha$ is well-defined owing to the first condition
on $\indfam{\alpha_k}{k \in \Nat}$.
The second condition is satisfiable, but it is not satisfied by all
$\indfam{\alpha_k}{k \in \Nat}$ satisfying the first condition.

The basic idea underlying the encoding of \iiiCNF-formulas as sequences
of Boolean values is as follows:
\begin{iteml}
\item
if $n = \ndisj(k)$ for some $k \in \Nat$, then the input of $\iiiSATC_n$
consists of one Boolean value for each disjunction of at most three
literals from the set $\set{v_1,\Not v_1,\ldots,v_k,\Not v_k}$;
\item
each Boolean value indicates whether the corresponding disjunction
occurs in the encoded \iiiCNF-formula;
\item
if $\ndisj(k) < n < \ndisj(k+1)$ for some $k \in \Nat$, then only the
first $\ndisj(k)$ Boolean values form part of the encoding.
\end{iteml}

For each $n \in \Nat$, $\funct{\iiiSATC_n}{\Bool^n}{\Bool}$ is
defined as follows:
\begin{iteml}
\item
if $n = \ndisj(k)$ for some $k \in \Nat$:
\begin{ldispl}
\iiiSATC_n\left(b_1,\ldots,b_n\right) = \True
\;\;\;\mathrm{iff}\;\;\;
\displaystyle
\AND{i \in \left[1,n\right] \;\mathrm{s.t.}\; b_i = \True}
 \OR{} \alpha_k(i)
\;\;\mathrm{is\;satisfiable}\;,
\end{ldispl}%
where $k$ is such that $n = \ndisj(k)$;
\item
if $\ndisj(k) < n < \ndisj(k+1)$ for some $k \in \Nat$:
\begin{ldispl}
\iiiSATC_n\left(b_1,\ldots,b_n\right) =
\iiiSATC_{\ndisj(k)}\left(b_1,\ldots,b_{\ndisj(k)}\right)\;,
\end{ldispl}%
where $k$ is such that $\ndisj(k) < n < \ndisj(k+1)$.
\end{iteml}

Because $\indfam{\alpha_k}{k \in \Nat}$ satisfies the condition that
${\alpha_i}^{-1}(\alpha_{i+1}(j)) = j$ for all $i \in \Nat$ and
$j \in [1,\ndisj(i)]$, we have for each $n \in \Nat$, for all
$b_1,\ldots,b_n \in \Bool$:
\begin{ldispl}
\iiiSATC_n\left(b_1,\ldots,b_n\right) =
\iiiSATC_{n+1}\left(b_1,\ldots,b_n,\False\right)\;.
\end{ldispl}%
In other words, for each $n \in \Nat$, $\iiiSATC_{n+1}$ can in essence
handle all inputs that $\iiiSATC_n$ can handle.
This means that $\indfam{\iiiSATC_n}{n \in \Nat}$ converges to the
unique function $\funct{\iiiSATC}{\seqof{\Bool}}{\Bool}$ such that for
each $n \in \Nat$, for each $w \in \Bool^n$,
$\iiiSATC(w) = \iiiSATC_n(w)$.

\iiiSATC\ is meant to correspond to \iiiSAT.
Therefore, the following theorem does not come as a surprise.
Notice that we identify in this theorem the Boolean function family
$\iiiSATC = \indfam{\iiiSATC_n}{n \in \Nat}$ with the unique function
$\funct{\iiiSATC}{\seqof{\Bool}}{\Bool}$ such that for each
$n \in \Nat$, for each $w \in \Bool^n$, $\iiiSATC(w) = \iiiSATC_n(w)$.
\begin{theorem}
\label{theorem-3SATC-is-NP-compl}
$\iiiSATC$ is \NPT-complete.
\end{theorem}
\begin{proof}
$\iiiSATC$ is \NPT-complete iff \iiiSATC\ is in \NPT\ and \iiiSATC\ is
\NPT-hard.
Because \iiiSAT\ is \NPT-complete, it is sufficient to prove that
\iiiSATC\ is polynomial-time Karp reducible to \iiiSAT\ and
\iiiSAT\ is polynomial-time Karp reducible to \iiiSATC, respectively.
In the rest of the proof, $\alpha$ is defined as above.

Take the function $f$ from $\seqof{\Bool}$ to the set of all
\iiiCNF-formulas containing the variables $v_1,\ldots,v_k$ for some
$k \in \Nat$ that is defined by
$f(b_1,\ldots,b_n) = \linebreak[2]
 \AND{i \in [1,\max \set{\ndisj(k) \where \ndisj(k) \leq n}]
   \;\mathrm{s.t.}\; b_i = \True}
  \OR{} \alpha(i)$.
Then we have that
$\iiiSATC(b_1,\ldots,b_n) = \iiiSAT(f(b_1,\ldots,b_n))$.
It remains to show that $f$ is polynomial-time computable.
To compute $f(b_1,\ldots,b_n)$, $\alpha$ has to be computed for a number
of times that is not greater than $n$ and $\alpha$ is computable in time
polynomial in $n$.
Hence, $f$ is polynomial-time computable.

Take the unique function $g$ from the set of all \iiiCNF-formulas
containing the variables $v_1,\ldots,v_k$ for some $k \in \Nat$ to
$\seqof{\Bool}$ such that for all \iiiCNF-formulas $\phi$ containing the
variables $v_1,\ldots,v_k$ for some $k \in \Nat$, $f(g(\phi)) = \phi$
and there exists no $w \in \seqof{\Bool}$ shorter than $g(\phi)$ such
that $f(w) = \phi$.
We have that $\iiiSAT(\phi) = \iiiSATC(g(\phi))$.
It remains to show that $g$ is polynomial-time computable.
To compute $g(\phi)$, where $l$ is the size of $\phi$, $\alpha$ has to
be computed for each clause a number of times that is not greater than
$\ndisj(l)$ and $\alpha$ is computable in time polynomial in
$\ndisj(l)$.
Moreover, $\phi$ contains at most $l$ clauses.
Hence, $g$ is polynomial-time computable.
\qed
\end{proof}

Before we turn to the non-uniform super-polynomial complexity
hypothesis, we touch lightly on the choice of the family of bijections
in the definition of $\iiiSATC$.
It is easy to see that the choice is not essential.
Let $\iiiSATC'$ be the same as \iiiSATC, but based on another family of
bijections, say $\indfam{\alpha'_n}{n \in \Nat}$, and let, for each
$i \in \Nat$, for each $j \in \left[1,\ndisj(i)\right]$,
$b'_j = b_{{\alpha_i}^{-1}(\alpha'_i(j))}$.
Then:
\begin{iteml}
\item
if $n = \ndisj(k)$ for some $k \in \Nat$:
\begin{ldispl}
\iiiSATC_n\left(b_1,\ldots,b_n\right) =
\iiiSATC'_n\left(b'_1,\ldots,b'_n\right)\;;
\end{ldispl}%
\item
if $\ndisj(k) < n < \ndisj(k+1)$ for some $k \in \Nat$:
\begin{ldispl}
\iiiSATC_n\left(b_1,\ldots,b_n\right) =
\iiiSATC'_n\left(b'_1,\ldots,b'_{\ndisj(k)},
                 b_{\ndisj(k)+1},\ldots,b_n\right)\;,
\end{ldispl}%
where $k$ is such that $\ndisj(k) < n < \ndisj(k+1)$.
\end{iteml}
This means that the only effect of another family of bijections is
another order of the relevant arguments.

The \emph{non-uniform super-polynomial complexity hypothesis} is the
following hypothesis:
\begin{hypothesis}
\label{hypothesis-basic}
$\iiiSATC \notin \PLIS$.
\end{hypothesis}

$\iiiSATC \notin \PLIS$ expresses in short that there does not exist a
polynomial function $\funct{h}{\Nat}{\Nat}$ such that for all
$n \in \Nat$ there exists an $X \in \ISbr$ such that $X$ computes
$\iiiSATC_n$ and $\psize(X) \leq h(n)$.
This corresponds with the following informal formulation of the
non-uniform super-polynomial complexity hypothesis:
\begin{quote}
the lengths of the shortest instruction sequences that compute the \\
Boolean functions $\iiiSATC_n$ are not bounded by a polynomial in $n$.
\end{quote}

The statement that Hypothesis~\ref{hypothesis-basic} is a counterpart of
the conjecture that $\iiiSAT \notin \PTpoly$ is made rigorous in the
following theorem.
\begin{theorem}
\label{theorem-equiv-hypotheses}
$\iiiSATC \notin \PLIS$ is equivalent to\, $\iiiSAT \notin \PTpoly$.
\end{theorem}
\begin{proof}
This follows immediately from Theorems~\ref{theorem-PLIS-is-PTpoly}
and~\ref{theorem-3SATC-is-NP-compl} and the fact that \iiiSAT\ is
\NPT-complete.
\qed
\end{proof}

\section{Splitting of Instruction Sequences and Threads}
\label{sect-splitting}

The instruction sequences considered in \PGA\ are sufficient to define a
counterpart of \PTpoly, but not to define a counterpart of \NPTpoly.
For a counterpart of \NPTpoly, we introduce in this section an extension
of \PGA\ that allows for single-pass instruction sequences to split.
We also introduce an extension of \BTA\ with a behavioural counterpart
of instruction sequence splitting that is reminiscent of thread forking.
First, we extend \PGA\ with instruction sequence splitting.

It is assumed that a fixed but arbitrary countably infinite set
$\BoolPar$ of \emph{Boolean parameters} has been given.
Boolean parameters are used to set up a simple form of parameterization
for single-pass instruction sequences.

\PGAsplit\ is \PGA\ with built-in basic instructions for instruction
sequence splitting.
In \PGAsplit, the following basic instructions belong to $\BInstr$:
\begin{iteml}
\item
for each $\bp \in \BoolPar$, a \emph{splitting instruction}
$\split{\bp}$;
\item
for each $\bp \in \BoolPar$, a \emph{direct replying instruction}
$\reply{\bp}$.
\end{iteml}

On execution of the instruction sequence $\ptst{\split{\bp}} \conc X$,
the primitive instruction $\ptst{\split{\bp}}$ brings about concurrent
execution of the instruction sequence $X$ with the Boolean parameter
$\bp$ instantiated to $\True$ and the instruction sequence
$\fjmp{2} \conc X$ with the Boolean parameter $\bp$ instantiated to
$\False$.
The case where $\ptst{\split{\bp}}$ is replaced by $\ntst{\split{\bp}}$
differs in the obvious way, and likewise the case where
$\ptst{\split{\bp}}$ is replaced by $\split{\bp}$.

On execution of the instruction sequence $\ptst{\reply{\bp}} \conc X$,
the primitive instruction $\ptst{\reply{\bp}}$ brings about execution
of the instruction sequence $X$ if the value taken by the Boolean
parameter $\bp$ is $\True$ and execution of the instruction sequence
$\fjmp{2} \conc X$ if the value taken by the Boolean parameter $\bp$
is $\False$.
The case where $\ptst{\reply{\bp}}$ is replaced by $\ntst{\reply{\bp}}$
differs in the obvious way, and likewise the case where
$\ptst{\reply{\bp}}$ is replaced by $\reply{\bp}$.

The axioms of \PGAsplit\ are the same as the axioms of \PGA.
The thread extraction operation for closed \PGAsplit\ terms in which the
repetition operator does not occur is defined as for closed \PGA\ terms
in which the repetition operator does not occur.
However, in the presence of the additional instructions of \PGAsplit,
the intended behaviour of the instruction sequence denoted by a closed
term $P$ is not $\extr{P}$.
In the notation of the extension of \BTA\ introduced below, the intended
behaviour is described by $\csi{\seq{\extr{P}}}$.

Henceforth, we will write $\FSIS$ for the set of all instruction
sequences that can be denoted by closed \PGAsplit\ terms in which the
repetition operator does not occur.
Moreover, we will write $\SISbr$ for the set of all instruction
sequences from $\FSIS$ in which all primitive instructions, with the
exception of jump instructions and the termination instruction, contain
only basic instructions from the set
\begin{ldispl}
\set{f.\getbr \where f \in \Fociin} \union
\set{\outbr.\setbr{\True}} \union
\set{\split{\bp}, \reply{\bp} \where \bp \in \BoolPar}\;.
\end{ldispl}%
Notice that no auxiliary registers are used in instruction sequences
from $\SISbr$ and that the basic instruction $\outbr.\setbr{\False}$
does not occur in instruction sequences from $\SISbr$.

In the remainder of this section, we extend \BTA\ with a mechanism for
multi-threading that supports thread splitting, the behavioural
counterpart of instruction sequence splitting.
This extension is entirely tailored to the behaviours of the instruction
sequences that can be denoted by closed \PGAsplit\ terms.

It is assumed that the collection of threads to be interleaved takes
the form of a sequence of threads, called a \emph{thread vector}.

The interleaving of threads is based on the simplest deterministic
interleaving strategy treated in~\cite{BM04c}, namely cyclic
interleaving, but any other plausible deterministic interleaving
strategy would be appropriate for our purpose.%
\footnote
{Fairness of the strategy is not an issue because the behaviours of the
 instruction sequences that can be denoted by closed \PGAsplit\ terms
 are finite threads.
 However, deadlock of one thread in the thread vector should not prevent
 others to proceed.}
Cyclic interleaving basically operates as follows: at each stage of the
interleaving, the first thread in the thread vector gets a turn to
perform a basic action and then the thread vector undergoes cyclic
permutation.
We mean by cyclic permutation of a thread vector that the first thread
in the thread vector becomes the last one and all others move one
position to the left.
If one thread in the thread vector deadlocks, the whole does not
deadlock till all others have terminated or deadlocked.

We introduce the additional sort $\TV$ of \emph{thread vectors}.
To build terms of sort $\Thr$, we introduce the following additional
operators:
\begin{iteml}
\item
the unary \emph{cyclic interleaving} operator
$\funct{\csiop {}}{\TV}{\Thr}$;
\item
the unary \emph{deadlock at termination} operator
$\funct{\stdop}{\Thr}{\Thr}$;
\item
for each $\bp \in \BoolPar$ and $b \in \set{\True,\False}$,
the unary \emph{parameter instantiation} operator
$\funct{\instaop{\bp}{b}}{\Thr}{\Thr}$;
\item
for each $\bp \in \BoolPar$,
the two binary \emph{postconditional composition} operators
$\funct{\pcc{\ph}{\split{\bp}}{\ph}}{\Thr \x \Thr}{\Thr}$ and
$\funct{\pcc{\ph}{\reply{\bp}}{\ph}}{\Thr \x \Thr}{\Thr}$.
\end{iteml}
To build terms of sort $\TV$, we introduce the following constants and
operators:
\begin{itemize}
\item
the \emph{empty thread vector} constant $\const{\emptyseq}{\TV}$;
\item
the \emph{singleton thread vector} operator
$\funct{\seq{\ph}}{\Thr}{\TV}$;
\item
the \emph{thread vector concatenation} operator
$\funct{\ph \concat \ph}{\TV \x \TV}{\TV}$.
\end{itemize}
Throughout the paper, we assume that there are infinitely many variables
of sort $\TV$, including $\alpha$.

For an operational intuition, $\split{\bp}$ can be considered a thread
splitting action: when encountering $\pcc{p}{\split{\bp}}{q}$ at some
stage of interleaving, this thread is split into two threads, namely $p$
with the Boolean parameter $\bp$ instantiated to $\True$ and $q$ with
the Boolean parameter $\bp$ instantiated to $\False$.
For an operational intuition, $\reply{\bp}$ can be considered a direct
replying action: on performing $\reply{\bp}$ the value taken by the
Boolean parameter $\bp$ is returned as reply value without any further
processing.

Intuitively, $\csi{\alpha}$ is the thread that results from cyclic
interleaving of the threads in the thread vector $\alpha$, covering the
above-mentioned splitting of a thread in the thread vector into two
threads.
This splitting involves instantiation of Boolean parameters in threads.
Intuitively, $\insta{\bp}{b}{p}$ is the thread that results from
instantiating the Boolean parameter $\bp$ to $b$ in thread $p$.
In the event of deadlock of one thread in the thread vector, the whole
deadlocks only after all others have terminated or deadlocked.
The auxiliary operator $\stdop$ is introduced to describe this fully
precise.
Intuitively, $\std{p}$ is the thread that results from turning
termination into deadlock in $p$.

The axioms for cyclic interleaving with thread splitting, deadlock at
termination, and parameter instantiation are given in
Tables~\ref{axioms-csi-split}, \ref{axioms-std} and~\ref{axioms-bpar}.%
\begin{table}[!tb]
\caption{Axioms for cyclic interleaving with thread splitting}
\label{axioms-csi-split}
\begin{eqntbl}
\begin{axcol}
\csi{\emptyseq} = \Stop                                & \axiom{CSI1} \\
\csi{\seq{\Stop}\concat \alpha} = \csi{\alpha}         & \axiom{CSI2} \\
\csi{\seq{\DeadEnd} \concat \alpha} = \std{\csi{\alpha}}
                                                       & \axiom{CSI3} \\
\csi{\seq{\Tau \bapf x}\concat \alpha} =
\Tau \bapf \csi{\alpha \concat \seq{x}}                & \axiom{CSI4} \\
\csi{\seq{\pcc{x}{a}{y}}\concat \alpha} =
\pcc{\csi{\alpha \concat \seq{x}}}{a}
    {\csi{\alpha \concat \seq{y}}}                     & \axiom{CSI5} \\
\csi{\seq{\pcc{x}{\split{\bp}}{y}}\concat \alpha} =
\Tau \bapf \csi{\alpha \concat
                \seq{\insta{\bp}{\True}{x}} \concat
                \seq{\insta{\bp}{\False}{y}}}          & \axiom{CSI6} \\
\csi{\seq{\pcc{x}{\reply{\bp}}{y}}\concat \alpha} = \std{\csi{\alpha}}
                                                       & \axiom{CSI7}
\end{axcol}
\end{eqntbl}
\end{table}
\begin{table}[!tb]
\caption{Axioms for deadlock at termination}
\label{axioms-std}
\begin{eqntbl}
\begin{axcol}
\std{\Stop} = \DeadEnd                                 & \axiom{S2D1} \\
\std{\DeadEnd} = \DeadEnd                              & \axiom{S2D2} \\
\std{\Tau \bapf x} = \Tau \bapf \std{x}                & \axiom{S2D3} \\
\std{\pcc{x}{a}{y}} = \pcc{\std{x}}{a}{\std{y}}        & \axiom{S2D4} \\
\std{\pcc{x}{\split{\bp}}{y}} = \pcc{\std{x}}{\split{\bp}}{\std{y}}
                                                       & \axiom{S2D5} \\
\std{\pcc{x}{\reply{\bp}}{y}} = \pcc{\std{x}}{\reply{\bp}}{\std{y}}
                                                       & \axiom{S2D6}
\end{axcol}
\end{eqntbl}
\end{table}
\begin{table}[!tb]
\caption{Axioms for parameter instantiation}
\label{axioms-bpar}
\begin{eqntbl}
\begin{saxcol}
\insta{\bp}{b}{\Stop} = \Stop                        & & \axiom{BPI1} \\
\insta{\bp}{b}{\DeadEnd} = \DeadEnd                  & & \axiom{BPI2} \\
\insta{\bp}{b}{\Tau \bapf x} = \Tau \bapf \insta{\bp}{b}{x}
                                                     & & \axiom{BPI3} \\
\insta{\bp}{b}{\pcc{x}{a}{y}} =
\pcc{\insta{\bp}{b}{x}}{a}{\insta{\bp}{b}{y}}        & & \axiom{BPI4} \\
\insta{\bp}{b}{\pcc{x}{\split{\bp'}}{y}} =
\pcc{\insta{\bp}{b}{x}}{\split{\bp'}}{\insta{\bp}{b}{y}}
                                  & \mif \bp \neq \bp' & \axiom{BPI5} \\
\insta{\bp}{b}{\pcc{x}{\split{\bp}}{y}} = \DeadEnd   & & \axiom{BPI6} \\
\insta{\bp}{b}{\pcc{x}{\reply{\bp'}}{y}} =
\pcc{\insta{\bp}{b}{x}}{\reply{\bp'}}{\insta{\bp}{b}{y}}
                                  & \mif \bp \neq \bp' & \axiom{BPI7} \\
\insta{\bp}{\True}{\pcc{x}{\reply{\bp}}{y}} =
\Tau \bapf \insta{\bp}{\True}{x}                     & & \axiom{BPI8} \\
\insta{\bp}{\False}{\pcc{x}{\reply{\bp}}{y}} =
\Tau \bapf \insta{\bp}{\False}{y}                    & & \axiom{BPI9}
\end{saxcol}
\end{eqntbl}
\end{table}
In these tables, $a$ stands for an arbitrary action from $\BAct$.
With the exception of CSI7 and BPI6, the axioms simply formalize the
informal explanations given above.
Axiom CSI7 expresses that deadlock takes place when $\reply{\bp}$
ought to be performed next but $\bp$ is an uninstantiated Boolean
parameter.
Axiom BPI6 expresses that deadlock takes place when $\split{\bp}$
ought to be performed next but $\bp$ is an instantiated Boolean
parameter.
To be fully precise, we should give axioms concerning the constants and
operators to build terms of the sort $\TV$ as well.
We refrain from doing so because the constants and operators concerned
are the usual ones for sequences.

To simplify matters, we will henceforth take the set
$\set{\bpar{i} \where i \in \Natpos}$ for the set $\BoolPar$ of Boolean
parameters.

\section{The Complexity Class \PLSIS}
\label{sect-PLSIS}

In this section, we introduce the class \PLSIS\ of all Boolean function
families that can be computed by polynomial-length instruction sequences
from $\SISbr$.

Let $n \in \Nat$, let $\funct{f}{\Bool^n}{\Bool}$, and let
$X \in \SISbr$.
Then $X$ \emph{splitting computes} $f$ if for all
$b_1,\ldots,b_n \in \Bool$:
\begin{ldispl}
( \ldots
 (\csi{\seq{\extr{X}}}
   \useop{\inbr{1}} \BR_{b_1}) \ldots \useop{\inbr{n}} \BR_{b_n})
 \applyop{\outbr} \BR_\False =
\BR_{f(b_1,\ldots,b_n)}\;.
\end{ldispl}%

$\PLSIS$ is the class of all Boolean function families
$\indfam{f_n}{n \in \Nat}$ that satisfy:
\begin{quote}
there exists a polynomial function $\funct{h}{\Nat}{\Nat}$ such that for
all $n \in \Nat$ there exists an $X \in \SISbr$ such that $X$ splitting
computes $f_n$ and $\psize(X) \leq h(n)$.
\end{quote}

A question that arises is how \PLIS\ and \PLSIS\ are related.
It happens that \PLIS\ is included in \PLSIS.
\begin{theorem}
\label{theorem-PLIS-incl-PLSIS}
$\PLIS \subseteq \PLSIS$.
\end{theorem}
\begin{proof}
Suppose that $\indfam{f_n}{n \in \Nat}$ in $\PLIS$.
Let $n \in \Nat$, and let $X \in \ISbr$ be such that $X$ computes $f_n$
and $\psize(X)$ is polynomial in $n$.
Assume that the basic instruction $\outbr.\setbr{\False}$ does not occur
in $X$.
By Theorem~\ref{theorem-comput-output}, this assumption can be made
without loss of generality.
Then an $Y \in \SISbr$ such that $Y$ splitting computes $f_n$ and
$\psize(Y)$ is polynomial in $n$ can be obtained from $X$ as described
below.

Suppose that $X = u_1 \conc \ldots \conc u_k$.
Let $X' \in \ISbr$ be obtained from $u_1 \conc \ldots \conc u_k$ as
follows:
\begin{enuml}
\item
stop if there exists no $i \in [1,k]$ such that
$u_i \equiv \ntst{\auxbr{j}.\setbr{\True}}$ or
$u_i \equiv \ptst{\auxbr{j}.\setbr{\False}}$ for some $j \in \Natpos$;
\item
find the least $i \in [1,k]$ such that
$u_i \equiv \ntst{\auxbr{j}.\setbr{\True}}$ or
$u_i \equiv \ptst{\auxbr{j}.\setbr{\False}}$ for some $j \in \Natpos$;
\item
if $u_i \equiv \ntst{\auxbr{j}.\setbr{\True}}$ for some $j \in \Natpos$,
then replace $u_i$ by $\ptst{\auxbr{j}.\setbr{\True}} \conc \fjmp{2}$;
\item
if $u_i \equiv \ptst{\auxbr{j}.\setbr{\False}}$ for some
$j \in \Natpos$,\,
then replace $u_i$ by $\ntst{\auxbr{j}.\setbr{\False}} \conc \fjmp{2}$;
\item
for each $i' \in [1,k]$, replace $u_{i'}$ by $\fjmp{l + 1}$ if
$u_{i'} \equiv \fjmp{l}$ and $i' < i < i' + l$;
\item
repeat the preceding steps for the resulting instruction sequence.
\end{enuml}
Now, suppose that $X' = u'_1 \conc \ldots \conc u'_{k'}$.
Let $Y \in \SISbr$ be obtained from $u'_1 \conc \ldots \conc u'_{k'}$ as
follows:
\begin{enuml}
\item
stop if there exists no $i \in [1,k']$ such that
$u'_i \equiv \auxbr{j}.\setbr{b}$ or
$u'_i \equiv \ptst{\auxbr{j}.\setbr{\hsp{-.035}\True}}$ or
$u'_i \equiv \ntst{\auxbr{j}.\setbr{\False}}$
for some $j \in \Natpos$ and $b \in \Bool$;
\item
find the greatest $i \in [1,k']$ such that
$u'_i \equiv \auxbr{j}.\setbr{b}$ or
$u'_i \equiv \ptst{\auxbr{j}.\setbr{\True}}$\linebreak[2] or
$u'_i \equiv \ntst{\auxbr{j}.\setbr{\False}}$
for some $j \in \Natpos$ and $b \in \Bool$;
\item
find the unique $j \in \Natpos$ such that focus $\auxbr{j}$ occurs in
$u'_i$;
\item
find the least $j' \in \Natpos$ such that parameter $\bpar{j'}$ does not
occur in $u'_i \conc \ldots \conc u'_{k'}$;
\item
if $u'_i \equiv \auxbr{j}.\setbr{\True}$ or
$u'_i \equiv \ptst{\auxbr{j}.\setbr{\True}}$, then
replace $u'_i$ by $\ntst{\split{\bpar{j'}}} \conc \halt$;
\item
if $u'_i \equiv \auxbr{j}.\setbr{\False}$ or
$u'_i \equiv \ntst{\auxbr{j}.\setbr{\False}}$,\, then
replace $u'_i$ by $\ptst{\split{\bpar{j'}}} \conc \halt$;
\item
for each $i' \in [1,k']$, replace $u'_{i'}$ by $\fjmp{l+1}$ if
$u'_{i'} \equiv \fjmp{l}$ and $i' < i < i' + l$;
\item
for each $i' \in [i+1,k']$:
\begin{enuml}
\item
if $u'_{i'} \equiv \auxbr{j}.\getbr$, then replace $u'_{i'}$ by
$\reply{\bpar{j'}}$,
\item
if $u'_{i'} \equiv \ptst{\auxbr{j}.\getbr}$, then replace $u'_{i'}$ by
$\ptst{\reply{\bpar{j'}}}$,
\item
if $u'_{i'} \equiv \ntst{\auxbr{j}.\getbr}$, then replace $u'_{i'}$ by
$\ntst{\reply{\bpar{j'}}}$;
\end{enuml}
\item
repeat the preceding steps for the resulting instruction sequence.
\end{enuml}
It is easy to prove by induction on $k$ that the Boolean function
computed by $X$ and the Boolean function computed by $X'$ are the same,
and it is easy to prove by induction on $k'$ that the
Boolean function computed by $X'$ and the Boolean function splitting
computed by $Y$ are the same.
Moreover, it is easy to see that $\psize(Y) \leq 3 \mul \psize(X)$.
Hence, $\psize(Y)$ is also polynomial in $n$.
\qed
\end{proof}
The chances are that $\PLSIS \not\subseteq \PLIS$.
In Section~\ref{sect-feature-elim}, we will hypothesize this.

In Section~\ref{sect-hypothesis}, we have hypothesized that
$\iiiSATC \notin \PLIS$.
The question arises whether $\iiiSATC \in \PLSIS$.
This question can be answered in the affirmative.

\begin{theorem}
\label{theorem-3SATC-in-PLSIS}
$\iiiSATC \in \PLSIS$.
\end{theorem}
\begin{proof}
Let $n \in \Nat$,
let $k \in \Nat$ be the unique $k$ such that
$\ndisj(k) \leq n < \ndisj(k+1)$, and,
for each $b_1,\ldots,b_n \in \Bool$,
let $\phi_{b_1,\ldots,b_n}$ be the formula
$\AND{i \in [1,\ndisj(k)] \,\mathrm{s.t.}\,
  b_i = \True} \OR{} \alpha_k(i)$.
We have that $\iiiSATC_n(b_1,\ldots,b_n) = \True$ iff
$\phi_{b_1,\ldots,b_n}$ is satisfiable.
Let $\psi$ be the basic Boolean formula
$\AND{i \in [1,n]} (\Not v_{k+i} \Or \OR{} \alpha_k(i))$.
We have that $\phi_{b_{k+1},\ldots,b_{k+n}}(b_1,\ldots,b_k)$ iff
$\psi(b_1,\ldots,b_{k+n})$.
Let $X \in \ISbrna$ be such that the basic instruction
$\outbr.\setbr{\False}$ does not occur in $X$, $X$ computes the Boolean
function induced by $\psi$, and $\psize(X)$ is polynomial in $n$.
It follows from Theorem~\ref{theorem-comput-boolform} that such an $X$
exists.
Assume that instructions $\inbr{i}.\getbr$, $\ptst{\inbr{i}.\getbr}$,
and $\ntst{\inbr{i}.\getbr}$ with $i > k$ do not occur in $X$.
It is obvious that this assumption can be made without loss of
generality.
Let $Y \in \SISbr$ be the instruction sequence obtained from $X$ by
replacing, for each $i \in [1,k]$, all occurrences of the primitive
instructions $\inbr{i}.\getbr$, $\ptst{\inbr{i}.\getbr}$, and
$\ntst{\inbr{i}.\getbr}$ by the primitive instructions
$\reply{\bpar{i}}$, $\ptst{\reply{\bpar{i}}}$, and
$\ntst{\reply{\bpar{i}}}$, respectively, and
let $Z = \split{\bpar{1}} \conc \ldots \conc \split{\bpar{k}} \conc Y$.
We have that $Z \in \SISbr$, $Z$ splitting computes $\iiiSATC_n$, and
$\psize(Z)$ is polynomial in $n$.
Hence, $\iiiSATC \in \PLSIS$.
\qed
\end{proof}

Below we will define \PLSIS-completeness.
We would like to call it the counterpart of \NPTpoly-completeness in the
current setting, but the notion of \NPTpoly-completeness looks to be
absent in the literature on complexity theory.
The closest to \NPTpoly-completeness that we could find is
$p$-completeness for pD, a notion introduced in~\cite{SV85a}.
Like \NPT-completeness, \PLSIS-completeness will be defined in terms of
a reducibility relation.
Because \iiiSATC\ is closely related to \iiiSAT\ and
$\iiiSATC \in \PLSIS$, we expect \iiiSATC\ to be \PLSIS-complete.

Let $l,m,n \in \Nat$, and let
$\funct{f}{\Bool^n}{\Bool}$ and $\funct{g}{\Bool^m}{\Bool}$.
Then $f$ is \emph{length $l$ reducible} to $g$, written $f \llred{l} g$,
if there exist $\funct{h_1,\ldots,h_m}{\Bool^n}{\Bool}$ such that:
\begin{iteml}
\item
there exist $X_1,\ldots,X_m \in \ISbr$ such that $X_1,\ldots,X_m$
compute $h_1,\ldots,h_m$ and $\psize(X_1),\ldots,\psize(X_m) \leq l$;
\item
for all $b_1,\ldots,b_n \in \Bool$,
$f(b_1,\ldots,b_n) = g(h_1(b_1,\ldots,b_n),\ldots,h_m(b_1,\ldots,b_n))$.
\end{iteml}

Let $\indfam{f_n}{n \in \Nat}$ and $\indfam{g_n}{n \in \Nat}$ be
Boolean function families.
Then $\indfam{f_n}{n \in \Nat}$ is
\emph{non-uniform polynomial-length reducible} to
$\indfam{g_n}{n \in \Nat}$,
written $\indfam{f_n}{n \in \Nat} \plred \indfam{g_n}{n \in \Nat}$,
if there exists a polynomial function $\funct{q}{\Nat}{\Nat}$ such that:
\begin{iteml}
\item
for all $n \in \Nat$, there exist $l,m \in \Nat$ with
$l,m \leq q(n)$ such that $f_n \llred{l} g_m$.
\end{iteml}

Let $\indfam{f_n}{n \in \Nat}$ be a Boolean function family.
Then $\indfam{f_n}{n \in \Nat}$ is \emph{\PLSIS-complete} if:
\begin{iteml}
\item
$\indfam{f_n}{n \in \Nat} \in \PLSIS$;
\item
for all $\indfam{g_n}{n \in \Nat} \in \PLSIS$,
$\indfam{g_n}{n \in \Nat} \plred \indfam{f_n}{n \in \Nat}$.
\end{iteml}

The most important properties of non-uniform polynomial-length
reducibility and \PLSIS-completeness as defined above are stated in the
following two propositions.
\begin{proposition}
\label{prop-plred-props}
\mbox{}
\begin{enuml}
\item
if $\indfam{f_n}{n \in \Nat} \plred \indfam{g_n}{n \in \Nat}$ and
$\indfam{g_n}{n \in \Nat} \in \PLIS$,
then $\indfam{f_n}{n \in \Nat} \in \PLIS$;
\item
$\plred$ is reflexive and transitive.
\end{enuml}
\end{proposition}
\begin{proof}
Both properties follow immediately from the definition of $\plred$.
\qed
\end{proof}
\begin{proposition}
\label{prop-PLSIS-compl-props}
\mbox{}
\begin{enuml}
\item
if $\indfam{f_n}{n \in \Nat}$ is \PLSIS-complete and
$\indfam{f_n}{n \in \Nat} \in \PLIS$, then $\PLSIS = \PLIS$;
\item
if $\indfam{f_n}{n \in \Nat}$ is \PLSIS-complete,
$\indfam{g_n}{n \in \Nat} \in \PLSIS$ and
$\indfam{f_n}{n \in \Nat} \plred \indfam{g_n}{n \in \Nat}$,
then $\indfam{g_n}{n \in \Nat}$ is \PLSIS-complete.
\end{enuml}
\end{proposition}
\begin{proof}
The first property follows immediately from the definition of
\PLSIS-com\-pleteness, and the second property follows immediately from
the definition of \PLSIS-completeness and the transitivity of $\plred$.
\qed
\end{proof}
The properties stated in Proposition~\ref{prop-PLSIS-compl-props} make
\PLSIS-completeness as defined above adequate for our purposes.
In the following proposition, non-uniform polynomial-length reducibility
is related to polynomial-time Karp reducibility ($\ptred$).
\begin{proposition}
\label{prop-plred-psred}
Let $\indfam{f_n}{n \in \Nat}$ and $\indfam{g_n}{n \in \Nat}$ be the
Boolean function families, and let $f$ and $g$ be the unique functions
$\funct{f,g}{\seqof{\Bool}}{\Bool}$ such that for each $n \in \Nat$,
for each $w \in \Bool^n$, $f(w) = f_n(w)$ and $g(w) = g_n(w)$.
Then $f \ptred g$ only if
$\indfam{f_n}{n \in \Nat} \plred \indfam{g_n}{n \in \Nat}$.
\end{proposition}
\begin{proof}
This property follows immediately from the definitions of $\ptred$ and
$\plred$, the well-known fact that $\PT \subseteq \PTpoly$
(see e.g.~\cite{AB09a}, Remark~6.8) , and
Theorem~\ref{theorem-PLIS-is-PTpoly}.
\qed
\end{proof}
The property stated in Proposition~\ref{prop-plred-psred} allows for
results concerning polynomial-time Karp reducibility to be reused in the
current setting.

Now we turn to the anticipated \PLSIS-completeness of \iiiSATC.
\begin{theorem}
\label{theorem-3SATC-is-PLSIS-compl}
\iiiSATC\ is \PLSIS-complete.
\end{theorem}
\begin{proof}
By Theorem~\ref{theorem-3SATC-in-PLSIS}, we have that
$\iiiSATC \in \PLSIS$.
It remains to prove that for all $\indfam{f_n}{n \in \Nat} \in \PLSIS$,
$\indfam{f_n}{n \in \Nat} \plred \iiiSATC$.

Suppose that $\indfam{f_n}{n \in \Nat} \in \PLSIS$.
Let $n \in \Nat$, and let $X \in \SISbr$ be such that $X$ splitting
computes $f_n$ and $\psize(X)$ is polynomial in $n$.
Assume that $\outbr.\setbr{\True}$ occurs only once in $X$.
This assumption can be made without loss of generality: \linebreak[2]
multiple occurrences can always be eliminated by replacement by jump
instructions (on execution, instructions possibly following those
occurrences do not change the state of the Boolean register named
$\outbr$).
Suppose that $X = u_1 \conc \ldots \conc u_k$, and let $l \in [1,k]$ be
such that $u_l$ is either $\outbr.\setbr{\True}$,
$\ptst{\outbr.\setbr{\True}}$ or $\ntst{\outbr.\setbr{\True}}$.

We look for a transformation that gives, for each
$b_1,\ldots,b_n \in \Bool$, a Boolean formula $\phi_{b_1,\ldots,b_n}$
such that $f_n(b_1,\ldots,b_n) = \True$ iff $\phi_{b_1,\ldots,b_n}$ is
satisfiable.
Notice that, for fixed initial states of the Boolean registers named
$\inbr{1},\ldots,\inbr{n}$, it is possible that there exist several
execution paths through $X$ because of the split instructions that may
occur in $X$.
We have that $f_n(b_1,\ldots,b_n) = \True$ iff there exists an execution
path through $X$ that reaches $u_l$ if the initial states of the Boolean
registers named $\inbr{1},\ldots,\inbr{n}$ are $b_1,\ldots,b_n$,
respectively.
The existence of such an execution path corresponds to the
satisfiability of the Boolean formula
$v_1 \And v_l \And \AND{i \in [2,k]} (v_i \Iff \OR{j \in B(i)} v_j)$,
where, for each $i \in [2,k]$, $B(i)$ is the set of all $j \in [1,i]$
for which execution may proceed with $u_i$ after execution of $u_j$ if
the initial states of the Boolean registers named
$\inbr{1},\ldots,\inbr{n}$ are $b_1,\ldots,b_n$, respectively.
Let $\phi_{b_1,\ldots,b_n}$ be this Boolean formula.
Then $f_n(b_1,\ldots,b_n) = \True$ iff $\phi_{b_1,\ldots,b_n}$ is
satisfiable.

For some $m \in \Nat$, $\phi_{b_1,\ldots,b_n}$ still has to be
transformed into a $w_{b_1,\ldots,b_n} \in \Bool^m$ such that
$\phi_{b_1,\ldots,b_n}$ is satisfiable iff
$\iiiSATC_m(w_{b_1,\ldots,b_n}) = \True$.
We look upon this transformation as a composition of two
transformations: first $\phi_{b_1,\ldots,b_n}$ is transformed into a
\iiiCNF-formula $\psi_{b_1,\ldots,b_n}$ such that
$\phi_{b_1,\ldots,b_n}$ is satisfiable iff $\psi_{b_1,\ldots,b_n}$ is
satisfiable, and next, for some $m \in \Nat$, $\psi_{b_1,\ldots,b_n}$ is
transformed into a $w_{b_1,\ldots,b_n} \in \Bool^m$ such that
$\psi_{b_1,\ldots,b_n}$ is satisfiable iff
$\iiiSATC_m(w_{b_1,\ldots,b_n}) = \True$.

\sloppy
It is easy to see that the size of $\phi_{b_1,\ldots,b_n}$ is polynomial
in $n$ and that $\tup{b_1,\ldots,b_n}$ can be transformed into
$\phi_{b_1,\ldots,b_n}$ in time polynomial in $n$.
It is well-known that each Boolean formula $\psi$ can be transformed in
time polynomial in the size of $\psi$ into a \iiiCNF-formula $\psi'$,
with size and number of variables linear in the size of $\psi$, such
that $\psi$ is satisfiable iff $\psi'$ is satisfiable
(see e.g.~\cite{BDG88a}, Theorem~3.7).
Moreover, it is known from the proof of
Theorem~\ref{theorem-3SATC-is-NP-compl} that each \iiiCNF-formula $\phi$
can be transformed in time polynomial in the size of $\phi$ into a
$w \in \Bool^{\ndisj(l)}$, where $l$ is the number of variables in
$\phi$, such that $\iiiSAT(\phi) = \iiiSATC(w)$.
From these facts, and Proposition~\ref{prop-plred-psred}, it follows
easily that $\indfam{f_n}{n \in \Nat}$ is non-uniform polynomial-length
reducible to $\iiiSATC$.
\qed
\end{proof}

It happens that \PLSIS\ and \NPTpoly\ coincide.
\begin{theorem}
\label{theorem-PLSIS-is-NPTpoly}
$\PLSIS = \NPTpoly$.
\end{theorem}
\begin{proof}
It follows easily from the definitions concerned that $f \in \NPTpoly$
iff there exist a $k \in \Nat$ and a $g \in \PTpoly$ such that, for all
$w \in \seqof{\Bool}$:
\begin{ldispl}
f(w) = \True \Iff
\Exists{c \in \seqof{\Bool}}{|c| \leq {|w|}^k \And g(w,c) = \True}\;.
\end{ldispl}%
Below, we will refer to such a $g$ as a \emph{checking function} for
$f$.
We will first prove the inclusion $\NPTpoly \subseteq \PLSIS$ and then
the inclusion $\PLSIS \subseteq \NPTpoly$.

$\NPTpoly \subseteq \PLSIS$:
Suppose that $f \in \NPTpoly$.
Then there exists a checking function for $f$.
Let $g$ be a checking function for $f$, and
let $\indfam{g_n}{n \in \Nat}$ be the Boolean function family
corresponding to $g$.
Because $g \in \PTpoly$, we have by Theorem~\ref{theorem-PLIS-is-PTpoly}
that $\indfam{g_n}{n \in \Nat} \in \PLIS$.
This implies that, for all $n \in \Nat$, there exists an $X \in \ISbr$
such that $X$ computes $g_n$ and $\psize(X)$ is polynomial in $n$.
For each $n \in \Nat$, let $X_n$ be such an $X$.
Moreover, let $\indfam{f_n}{n \in \Nat}$ be the Boolean function family
corresponding to $f$.
For each $n \in \Nat$, there exists an $m \in \Nat$ such that a
$Z \in \SISbr$ can be obtained from $X_m$ in the way followed in the
proof of Theorem~\ref{theorem-PLIS-incl-PLSIS} such that $Z$ splitting
computes $f_n$ and $\psize(Z)$ is polynomial in $n$.
Hence, each Boolean function family in $\NPTpoly$ is also in $\PLSIS$.

$\PLSIS \subseteq \NPTpoly$:
Suppose that $\indfam{f_n}{n \in \Nat}$ in $\PLSIS$.
Then, for all $n \in \Nat$, there exists an $X \in \SISbr$ such that $X$
splitting computes $f_n$ and $\psize(X)$ is polynomial in $n$.
For each $n \in \Nat$, let $X_n$ be such an $X$.
Moreover, let $\funct{f}{\seqof{\Bool}}{\Bool}$ be the function
corresponding to $\indfam{f_n}{n \in \Nat}$.
Then a checking function $g$ for $f$ can be computed by a Turing machine
as follows: on an input of size $n$, it takes a binary description of
$X_n$ as advice and then simulates the execution of $X_n$ treating the
additional input as a description of the choices to make at each split.
It is easy to see that, under the assumption that instructions
$\split{\bpar{i}}$, $\ptst{\split{\bpar{i}}}$,
$\ntst{\split{\bpar{i}}}$, $\reply{\bpar{i}}$,
$\ptst{\reply{\bpar{i}}}$, $\ntst{\reply{\bpar{i}}}$
and $\fjmp{i}$ with $i > \psize(X_n)$ do not occur in $X_n$, the size of
the description of $X_n$ and the number of steps that it takes to
simulate the execution of $X_n$ are both polynomial in $n$.
It is obvious that we can make the assumption without loss of
generality.
Hence, each Boolean function family in $\PLSIS$ is also in $\NPTpoly$.
\qed
\end{proof}

A known result about classical complexity classes turns out to be a
corollary of
Theorems~\ref{theorem-PLIS-is-PTpoly}, \ref{theorem-3SATC-is-NP-compl},
\ref{theorem-3SATC-is-PLSIS-compl} and~\ref{theorem-PLSIS-is-NPTpoly}.
\begin{corollary}
\label{corollary-equiv-hypotheses}
$\NPT \not\subseteq \PTpoly$ is equivalent to\,
$\NPTpoly \not\subseteq \PTpoly$.
\end{corollary}

Notice that it is justified by Theorem~\ref{theorem-PLSIS-is-NPTpoly} to
regard the definition of \PLSIS-completeness given in this paper as a
definition of \NPTpoly-completeness in the setting of single-pass
instruction sequences and consequently to read
Theorem~\ref{theorem-3SATC-is-PLSIS-compl} as \iiiSATC\ is
\NPTpoly-complete.

\section{Super-polynomial Feature Elimination Complexity Hypotheses}
\label{sect-feature-elim}

In this section, we introduce three complexity hypotheses which concern
restrictions on the instruction sequences with which Boolean functions
are computed.

By Theorem~\ref{theorem-PLIS-incl-PLSIS}, we have that
$\PLIS \subseteq \PLSIS$.
We hypothesize that $\PLSIS \not\subseteq \PLIS$.
We can think of \PLIS\ as roughly obtained from \PLSIS\ by restricting
the computing instruction sequences to non-splitting  instruction
sequences.
This motivates the formulation of the hypothesis that
$\PLSIS \not\subseteq \PLIS$ as a feature elimination complexity
hypothesis.

The
\emph{first super-polynomial feature elimination complexity hypothesis}
is the following hypothesis:
\begin{hypothesis}
\label{hypothesis-feat-elim-1}
Let $\funct{\rho}{\SISbr}{\ISbr}$ be such that, for each $X \in \SISbr$,
$\rho(X)$ computes the same Boolean function as $X$.
Then $\psize(\rho(X))$ is not polynomially bounded in $\psize(X)$.
\end{hypothesis}

We can also think of complexity classes obtained from \PLIS\ by
restricting the computing instruction sequences further.
They can, for instance, be restricted to instruction sequences in which:
\begin{iteml}
\item
the primitive instructions $f.m$, $\ptst{f.m}$ and $\ntst{f.m}$ with
$f \in \Fociaux$ do not occur;
\item
for some fixed $k \in \Nat$, the jump instructions $\fjmp{l}$ with
$l > k$ do not occur;
\item
the primitive instructions $\outbr.\setbr{\False}$,
$\ptst{\outbr.\setbr{\False}}$ and $\ntst{\outbr.\setbr{\False}}$ do not
occur;
\item
multiple termination instructions do not occur.
\end{iteml}
Below we introduce two hypotheses that concern the first two of these
restrictions.

The
\emph{second super-polynomial feature elimination complexity hypothesis}
is the following hypothesis:
\begin{hypothesis}
\label{hypothesis-feat-elim-2}
Let $\funct{\rho}{\ISbr}{\ISbrna}$ be such that, for each $X \in \ISbr$,
$\rho(X)$ computes the same Boolean function as $X$.
Then $\psize(\rho(X))$ is not polynomially bounded in $\psize(X)$.
\end{hypothesis}
The
\emph{third super-polynomial feature elimination complexity hypothesis}
is the following hypothesis:
\begin{hypothesis}
\label{hypothesis-feat-elim-3}
Let $k \in \Nat$, and let $\funct{\rho}{\ISbrna}{\ISbrna}$ be such that,
for each $X \in \ISbrna$, $\rho(X)$ computes the same Boolean function as
$X$ and, for each jump instruction $\fjmp{l}$ occurring in $\rho(X)$,
$l \leq k$.
Then $\psize(\rho(X))$ is not polynomially bounded in $\psize(X)$.
\end{hypothesis}

These hypotheses motivate the introduction of subclasses of $\PLIS$.
For each $k,l \in \Nat$, \PLRIS{k}{l}\ is the class of all Boolean
function families $\indfam{f_n}{n \in \Nat}$ that satisfy:
\begin{quote}
there exists a polynomial function $\funct{h}{\Nat}{\Nat}$ such that
for all $n \in \Nat$ there exists an $X \in \ISbr$ such that:
\begin{iteml}
\item
$X$ computes $f_n$ and $\psize(X) \leq h(n)$;
\item
instructions $f.m$, $\ptst{f.m}$ and $\ntst{f.m}$ with $f = \auxbr{i}$
for some $i > k$ do not occur in $X$;
\item
instructions $\fjmp{i}$ with $i > l$ do not occur in $X$.
\end{iteml}
\end{quote}
Moreover, for each $k,l \in \Nat$,    \PLRIS{k}{*}\ is the class
$\Union{l \in \Nat} \PLRIS{k}{l}$, and \PLRIS{*}{l}\ is the class
$\Union{k \in \Nat} \PLRIS{k}{l}$.

The hypotheses formulated above, can also be expresses in terms of these
subclasses of \PLIS:
Hypotheses~\ref{hypothesis-feat-elim-1}, \ref{hypothesis-feat-elim-2},
and~\ref{hypothesis-feat-elim-3} are equivalent to
$\PLSIS \not\subseteq \PLIS$, $\PLIS \not\subseteq \PLRIS{0}{*}$, and
$\PLRIS{0}{*} \not\subseteq \PLRIS{0}{k}$ for all $k \in \Nat$,
respectively.

\section{Conclusions}
\label{sect-concl}

We have developed theory concerning non-uniform complexity based on the
simple idea that each $n$-ary Boolean function can be computed by a
single-pass instruction sequence that contains only instructions to read
and write the contents of Boolean registers, forward jump instructions,
and a termination instruction.

We have defined the non-uniform complexity classes \PLIS\ and \PLSIS,
counterparts of the classical non-uniform complexity classes \PTpoly\
and \NPTpoly, and the notion of \PLSIS-completeness using a non-uniform
reducibility relation.
We have shown that \PLIS\ and \PLSIS coincide with \PTpoly\ and
\NPTpoly.
This makes it clear that there are close connections between non-uniform
complexity theory based on single-pass instruction sequences and
non-uniform complexity theory based on Turing machines with advice or
Boolean circuits.
We have also shown that \iiiSATC, a problem closely related to \iiiSAT,
is both \NPT-complete and \PLSIS-complete.

Moreover, we have formulated a counterpart of the well-known complexity
theoretic conjecture that $\NPT \not\subseteq \PTpoly$ and three
complexity hypotheses which concern restrictions on the instruction
sequences used for computation.
The latter three hypotheses are intuitively appealing in the setting of
single-pass instruction sequences.
The first of these has a natural counterpart in the setting of Turing
machines with advice, but not in the setting of Boolean circuits.
The second and third of these appear to have no natural counterparts in
the settings of both Turing machines with advice and Boolean circuits.

The approaches to computational complexity based on loop
programs~\cite{MR67a}, straight-line programs~\cite{GLF77a}, and
branching programs~\cite{BDFP86a} appear to be the closest related to
the approach followed in this paper.

The notion of loop program is far from abstract or general: a loop
program consists of assignment statements and possibly nested loop
statements of a special kind.
To our knowledge, this notion is only used in the work presented
in~\cite{MR67a}.
That work is mainly concerned with upper bounds on the running time of
loop programs that can be determined syntactically.

The notion of straight-line program is relatively close to the notion of
single-pass instruction sequence: a straight-line program is a sequence
of steps, where in each step a language is generated by selecting an
element from an alphabet or by taking the union, intersection or
concatenation of languages generated in previous steps.
In other words, a straight-line program can be looked upon as a
single-pass instruction sequence without test and jump instructions,
with basic instructions which are rather distant from those usually
found.
In~\cite{GLF77a}, a complexity measure for straight-line programs is
introduced which is closely related to Boolean circuit size.
To our knowledge, the notion of straight-line program is only used in
the work presented in~\cite{GLF77a,BDG85a}.

The notion of branching program is actually a generalization of the
notion of decision tree from trees to graphs, so the term branching
program seems rather far-fetched.
However, width two branching programs bear a slight resemblance to
threads as considered in basic thread algebra.
Branching programs are related to non-uniform space complexity like
Boolean circuits are related to non-uniform time complexity.
Like the notion of Boolean circuit, the notion of branching program
looks to be lasting in complexity theory (see e.g.~\cite{Weg00a}).

\subsubsection*{Acknowledgements}
We thank anonymous referees of a previous submission of this paper for
remarks that have led to significant improvements of several proof
outlines.

\appendix

\section{Beyond Instruction Sequence Congruence}
\label{app-beyond-is-congr}

Instruction sequence equivalence is a congruence and the axioms of \PGA\
are complete for this congruence.
In this appendix, we show that there are interesting coarser congruences
for which additional axioms for can be devised.

It follows from the defining equations of thread extraction that
instruction sequences that are the same after removal of chains of
forward jumps in favour of single jumps exhibit the same behaviour.
Such instruction sequences are called structurally congruent.
The additional axioms for structural congruence in the case of \PGAfin\
are given in Table~\ref{axioms-struct-congr}.%
\begin{table}[!tb]
\caption{Axioms for structural congruence}
\label{axioms-struct-congr}
\begin{eqntbl}
\begin{eqncol}
\fjmp{n+1} \conc u_1 \conc \ldots \conc u_n \conc \fjmp{0} =
\fjmp{0} \conc u_1 \conc \ldots \conc u_n \conc \fjmp{0}     \\
\fjmp{n+1} \conc u_1 \conc \ldots \conc u_n \conc \fjmp{m} =
\fjmp{m+n+1} \conc u_1 \conc \ldots \conc u_n \conc \fjmp{m}
\end{eqncol}
\end{eqntbl}
\end{table}
In this table, $n$ and $m$ stand for arbitrary numbers from $\Nat$ and
$u_1$, \ldots, $u_n$ stand for arbitrary primitive instructions
from~$\PInstr$.

If we take
$\set{f.\getbr \where f \in \Fociin \union \Fociaux} \union
 \set{f.\setbr{b} \where f \in \Fociaux \union \set{\outbr} \And
                         b \in \set{\True,\False}}$
for the set $\BInstr$ of basic instructions, then certain instruction
sequences can be identified because they exhibit the same behaviour on
the intended interaction with Boolean register services.
Such instruction sequences are called behaviourally congruent.
The additional axioms for behavioural congruence in this case are given
in Table~\ref{axioms-behav-congr}.%
\begin{table}[!tb]
\caption{Axioms for behavioural congruence}
\label{axioms-behav-congr}
\begin{eqntbl}
\begin{eqncol}
\ptst{f.\setbr{\True}} = f.\setbr{\True}                     \\
\ntst{f.\setbr{\False}} = f.\setbr{\False}                   \\
\ntst{f.\setbr{\True}} \conc f.\setbr{\True} =
\fjmp{1} \conc f.\setbr{\True}                               \\
\ptst{f.\setbr{\False}} \conc f.\setbr{\False} =
\fjmp{1} \conc f.\setbr{\False}                              \\
\ntst{f.\setbr{\True}} \conc \fjmp{n+2} \conc \fjmp{n+2} \conc
u_1 \conc \ldots \conc u_n \conc f.\setbr{\True} =
{} \\ \qquad\;\,
\fjmp{1} \conc \fjmp{n+2} \conc \fjmp{n+2} \conc
u_1 \conc \ldots \conc u_n \conc f.\setbr{\True}             \\
\ptst{f.\setbr{\False}} \conc \fjmp{n+2} \conc \fjmp{n+2} \conc
u_1 \conc \ldots \conc u_n \conc f.\setbr{\False} =
{} \\ \qquad\;\,
\fjmp{1} \conc \fjmp{n+2} \conc \fjmp{n+2} \conc
u_1 \conc \ldots \conc u_n \conc f.\setbr{\False}
\end{eqncol}
\end{eqntbl}
\end{table}
In this table,
$f$ stands for an arbitrary focus from $\Fociaux \union \set{\outbr}$,
$n$ stands for an arbitrary number from $\Nat$, and
$u_1$, \ldots, $u_n$ stand for arbitrary primitive instructions from
$\PInstr$.

\section{Explicit Substitution for Linear-size Thread Extraction}
\label{app-expl-subst}

In this appendix, we show that the combinatorial explosions mentioned at
the end of Section~\ref{sect-PGA} can be eliminated if we add explicit
substitution to \BTA.
We write $\Var$ for the countably infinite set of variables of sort
$\Thr$.

The extension of \BTA\ with explicit substitution has the constants and
operators of \BTA\ and in addition:
\begin{iteml}
\item
for each $x \in \Var$, the \emph{substitution} operator
$ \funct{\subst{\ph}{x}}{\Thr}{\Thr}$.
\end{iteml}
The additional axioms are given in Table~\ref{axioms-subst}.%
\begin{table}[!tb]
\caption{Axioms for substitution operators}
\label{axioms-subst}
\begin{eqntbl}
\begin{saxcol}
\subst{p}{X} X = p                                    & & \axiom{ES1} \\
\subst{p}{X} Y = Y                & \mif X \not\equiv Y & \axiom{ES2} \\
\subst{p}{X} \Stop = \Stop                            & & \axiom{ES3} \\
\subst{p}{X} \DeadEnd = \DeadEnd                      & & \axiom{ES4} \\
\subst{p}{X} (\pcc{q}{a}{r}) =
\pcc{(\subst{p}{X} q)}{a}{(\subst{p}{X} r)}           & & \axiom{ES5}
\end{saxcol}
\end{eqntbl}
\end{table}
In this table, $X$ and $Y$ stand for arbitrary variables from $\Var$,
$p$, $q$ and $r$ stand for arbitrary terms of this extension of \BTA\
with explicit substitution, and $a$ stands for an arbitrary action from
$\BActTau$.

The \emph{size} of a term of the extension of \BTA\ with explicit
substitution is defined as follows:
\begin{ldispl}
\begin{eqncol}
\tsize(X) = 1\;, \\
\tsize(\Stop) = 1\;, \\
\tsize(\DeadEnd) = 1\;,
\end{eqncol}
\qquad\qquad
\begin{eqncol}
\tsize(\pcc{p}{a}{q}) = \tsize(p) + \tsize(q) + 1\;, \\
\tsize(\subst{p}{X} q) = \tsize(p) + \tsize(q) + 1\;.
\end{eqncol}
\end{ldispl}%

The following theorem states that linear-size thread extraction is
possible if explicit substitution is added to \BTA.
\begin{theorem}
\label{theorem-tsize}
There exists a function $\rho$ from the set of all closed \PGAfin\ terms
to the set of all closed terms of the extension of \BTA\ with explicit
substitution such that for all closed \PGAfin\ terms $P$,
$\extr{P} = \rho(P)$ and $\tsize(\rho(P)) \leq 4 \mul \psize(P) + 1$.
\end{theorem}
\begin{proof}
For each $i \in \Nat$, let $x_i \in \Var$.
Take $\rho$ as follows ($k,l > 0$):
\begin{ldispl}
\rho(u_1) = \extr{u_1}\;, \\
\rho(u_1 \conc \ldots \conc u_{k+1}) =
\subst{\extr{u_{k+1}}}{x_{k+1}}(\subst{\rho'_k(u_k)}{x_k} \ldots
                        (\subst{\rho'_1(u_1)}{x_1} x_1) \ldots )\;,
\end{ldispl}%
where, for each $i \in [1,k]$:
\begin{ldispl}
\begin{aeqns}
\rho'_i(\halt)    & = & \Stop\;, \\
\rho'_i(\fjmp{0}) & = & \DeadEnd\;, \\
\rho'_i(\fjmp{l}) & = & x_{i+l}\;,
\end{aeqns}
\qquad\qquad
\begin{aeqns}
\rho'_i(a)        & = & \pcc{x_{i+1}}{a}{x_{i+1}}\;, \\
\rho'_i(\ptst{a}) & = & \pcc{x_{i+1}}{a}{x_{i+2}}\;, \\
\rho'_i(\ntst{a}) & = & \pcc{x_{i+2}}{a}{x_{i+1}}\;.
\end{aeqns}
\end{ldispl}%
It is easy to prove by induction on $\psize(P)$ that for all closed
\PGAfin\ terms $P$,
$\extr{P} = \rho(P)$ and $\tsize(\rho(P)) \leq 4 \mul \psize(P) + 1$.
\qed
\end{proof}

\bibliographystyle{splncs03}
\bibliography{TA}

\end{document}